\documentclass[a4paper,10pt]{article}
\usepackage[utf8x]{inputenc}

\usepackage{amsmath,amssymb,amsfonts,amsthm,dsfont}

\usepackage{bbm}
\usepackage{hyperref}
\usepackage{yfonts}

\title{{On the virtual levels of positively projected massless Coulomb-Dirac operators}}
\author{Sergey Morozov\footnote{Mathematisches Institut, Ludwig-Maximilians-Universit\"at M\"unchen, Theresienstr. 39, 80333 Munich, Germany\newline $\mathtt{\qquad morozov@math.lmu.de,\ dmueller@math.lmu.de}$}\ \footnote{St. Petersburg State University, Universitetskaya nab. 7-9, 199034 St. Petersburg, Russia} and David M\"uller$^{*}$}
\date{}

\newtheorem{thm}{Theorem}
\newtheorem{lem}[thm]{Lemma}

\newtheorem{cor}[thm]{Corollary}

\let\Im\undefined
\DeclareMathOperator{\Im}{Im}
\let\Re\undefined
\DeclareMathOperator{\Re}{Re}

\DeclareMathOperator{\tr}{tr}
\DeclareMathOperator{\dd}{d\!}
\DeclareMathOperator{\rank}{rank}

\DeclareMathOperator{\loc}{loc}

\DeclareMathOperator{\Ltwolim}{\mathsf L^2-lim}
\DeclareMathOperator{\Ran}{ran}
\DeclareMathOperator{\sign}{sign}

\numberwithin{equation}{section}

\begin{document}

\maketitle

\begin{abstract}
Considering different self-adjoint realisations of positively projected massless Coulomb-Dirac operators we find out, under which conditions any negative perturbation, however small, leads to emergence of negative spectrum. We also prove some weighted Lieb-Thirring estimates for negative eigenvalues of such operators. In the process we find explicit spectral representations for all self-adjoint realisations of massless Coulomb-Dirac operators on the half-line.
\end{abstract}

%\begin{center} Mathematisches Institut, Ludwig-Maximilians-Universit\"at M\"unchen\\ Theresienstr. 39, 80333 Munich, Germany.\\ %\smallskip $\mathtt{ \underline{morozov@math.lmu.de},\ dmueller@math.lmu.de}$\end{center}

% \medskip
% \noindent{\bfseries Abstract:} We study something.
% 
% \medskip
% \noindent{\bfseries Mathematics Subject Classification (2010):} 35P15, 35Q40.
% 
% \medskip
% \noindent{\bfseries Keywords:}

\section{Introduction}

In a related paper \cite{MorozovMueller} the authors have obtained estimates of Cwikel-Lieb-Rozenblum and Lieb-Thirring types on the negative eigenvalues of the perturbed positively projected two-dimensional massless Coulomb-Dirac operator emerging in the study of graphene. For the critical value of the coupling constant the control on the number of eigenvalues failed, which naturally posed a question about the existence of a virtual level at zero, i.e. the situation when every non-trivial negative perturbation leads to emergence of negative spectrum.

Trying to resolve this question we arrived at the study of Coulomb-Dirac operators on the half-line $\mathbb R_+$ associated to the differential expression
\begin{align}\label{Dirac symbol}
 d^{\nu, \kappa} := \begin{pmatrix}
  -\nu/r & -\frac{\mathrm d}{\mathrm dr}- \frac{\kappa}{r}\\ \frac{\mathrm d}{\mathrm dr}- \frac{\kappa}{r} & -\nu/r
 \end{pmatrix}.
\end{align}
Here $\nu\in\mathbb R$ is the strength of the Coulomb potential (nuclear charge) and $\kappa\in\mathbb R$ is a parameter typically arising after the separation of angular motion in several dimensions (see \cite{Thaller}). Typically $\kappa$ takes integer or half-integer values, but we will not need this assumption.
Throughout the text we use the notation
\begin{equation*}
 \beta:= \sqrt{\kappa^2 -\nu^2}\in \overline{\mathbb R_+}\cup\mathrm i\mathbb R_+.
\end{equation*}
It turns out that for $\beta \geqslant 1/2$ the operator $D^{\nu, \kappa}$ defined by \eqref{Dirac symbol} on $\mathsf C_0^\infty(\mathbb R_+, \mathbb C^2)$ is essentially self-adjoint in $\mathsf L^2(\mathbb R_+, \mathbb C^2)$. Otherwise, there exist a one parameter family $D^{\nu, \kappa}_\theta$, $\theta\in [0, \pi)$, of self-adjoint extensions of $D^{\nu, \kappa}$ differing by boundary conditions at zero. This is the result of Theorem \ref{t: self-adjoint realisations}, which is essentially contained in \cite{Weidmann1971}.
We denote the set of triples $(\nu, \kappa, \theta)$ for which $D^{\nu, \kappa}_\theta$ is defined in Theorem \ref{t: self-adjoint realisations} by $\mathfrak M$.

In Theorem \ref{t:spectral representation} for every $(\nu, \kappa, \theta)\in \mathfrak M$ we obtain the spectral representation of $D^{\nu, \kappa}_\theta$, namely we find an explicit unitary operator $\mathcal U^{\nu, \kappa}_\theta: \mathsf L^2(\mathbb R_+, \mathbb C^2, \mathrm dr)\to \mathsf L^2(\mathbb R, \mathbb C, \mathrm dx)$ such that $\mathcal U^{\nu, \kappa}_\theta D^{\nu, \kappa}_\theta(\mathcal U^{\nu, \kappa}_\theta)^*$ is the operator of multiplication by the independent variable.
In particular, for all $(\nu, \kappa, \theta)\in \mathfrak M$ the spectrum of $D^{\nu, \kappa}_\theta$ is purely absolutely continuous, simple and coincides with $\mathbb R$.
The existence and general form of spectral representations for one-dimensional Dirac systems is already proved in Theorem 9.7 of \cite{WeidmannSpectralTheory}, which provides a construction of the spectral representation with respect to a matrix-valued measure given by an explicit formula. Proving that this measure is of rank one and is mutually absolutely continuous with respect to the Lebesgue measure requires considerable work. Our proof of Theorem \ref{t:spectral representation} is not based on a direct application of this general result, since such an application seems to be more involved then our approach tailored specifically to $D^{\nu, \kappa}_\theta$. 
A related result for Coulomb-Dirac operators with positive mass can be found in \cite{VoronovGitmanTyutin}.

Then we study the negative spectrum of perturbations of $D^{\nu, \kappa}_\theta$ restricted to its positive spectral subspace. Denoting for any self-adjoint operator $A$ and Borel subset $\mathcal I$ of $\mathbb R$ the corresponding spectral projector by $P_{\mathcal I}(A)$ we define $P_{\theta, \infty}^{\nu, \kappa} :=P_{[0, \infty)}(D^{\nu, \kappa}_\theta)$. In $P_{\theta, \infty}^{\nu, \kappa}\mathsf L^2(\mathbb R_+, \mathbb C^2)$ we consider the operator
\begin{align}\label{D positive part with V}
 D_{\theta, \infty}^{\nu, \kappa}(V) :=P_{\theta, \infty}^{\nu, \kappa}(D^{\nu, \kappa}_\theta -V)P_{\theta, \infty}^{\nu, \kappa}
\end{align}
with $V$ being an operator of multiplication by a measurable Hermitian $2\times 2$-matrix-function.
For $2\times 2$ matrices their norms $\|\cdot\|_{\mathbb C^{2\times 2}}$, absolute values $|\cdot|$ and positive parts $(\cdot)_+$ are defined in the spectral sense.
We are, in particular, interested in the role which the choice of a self-adjoint realisation plays for the existence of a virtual level.

In Theorem \ref{t: the alternative} we obtain three types of results concerning virtual levels:
The parameter set $\mathfrak M$ is a disjoint union of $\mathfrak M_{\mathrm I}$, $\mathfrak M_{\mathrm{II}}$ and $\mathfrak M_{\mathrm{III}}$ such that
\begin{enumerate}
 \item[I.] For $(\nu, \kappa, \theta) \in\mathfrak M_{\mathrm I}$ the operator $D_{\theta, \infty}^{\nu, \kappa}(0)$ has a virtual level at zero, i.e. there exists a measurable function $A_\theta^{\nu, \kappa}: \mathbb R_+\to \mathbb C^2$ vanishing almost nowhere such that for any $V$ satisfying
\begin{align}\label{VL via A}
 \int_0^\infty\big\langle A_\theta^{\nu, \kappa}(r), V(r)A_\theta^{\nu, \kappa}(r)\big\rangle_{\mathbb C^{2}}\mathrm dr >0
\end{align}
and any $\alpha >0$ the operator $D_{\theta, \infty}^{\nu, \kappa}(\alpha V)$ has non-empty negative spectrum. Condition \eqref{VL via A} is satisfied for all $V \geqslant 0$ which are positive definite on some sets of positive Lebesgue measure.
 \item[II.] For $(\nu, \kappa, \theta) \in\mathfrak M_{\mathrm{II}}$ and $q >1$ there exist weight functions $W_{\theta, q}^{\nu, \kappa}: \mathbb R_+\to \overline{\mathbb R_+}$ such that the inequality
\begin{align}\label{integral estimate for beta>0}
 \rank P_{(-\infty, 0)}\big(D_{\theta, \infty}^{\nu, \kappa}(V)\big) \leqslant \int_0^\infty \big\|V_+(r)\big\|^q_{\mathbb C^{2\times2}}W^{\nu, \kappa}_{\theta, q}(r)\,\mathrm dr
\end{align}
holds. Consequently, for any $V$ for which the right hand side of \eqref{integral estimate for beta>0} is finite, the operator $D_{\theta, \infty}^{\nu, \kappa}(\alpha V)$ has no negative spectrum provided $|\alpha|$ is small enough.
\item[III.] For $(\nu, \kappa, \theta) \in\mathfrak M_{\mathrm{III}}$ and $V_+\in \mathsf L^\infty(\mathbb R_+, \mathbb C^{2\times2})$ the estimate
\begin{align}\begin{split}\label{no VL for beta=0}
 &\rank P_{(-\infty, 0)}\big(D_{\theta, \infty}^{\nu, \kappa}(V)\big)\\ &\leqslant K^{\nu, \kappa}\int_0^\infty \big\|V_+(r)\big\|_{\mathbb C^{2\times2}}\Big(\ln^2(\mathrm e^{\tan\theta}r) +\ln^2\big(\mathrm e +2r\|V_+\|_{\mathsf L^\infty(\mathbb R_+, \mathbb C^{2\times2})}\big)\Big)\,\mathrm dr
\end{split}\end{align}
holds with a finite constant $K^{\nu, \kappa}$.
Again, as soon as the right hand side of \eqref{no VL for beta=0} is finite, the operator $D_{\theta, \infty}^{\nu, \kappa}(\alpha V)$ has no negative spectrum provided $|\alpha|$ is small enough.
\end{enumerate}
The questions of existence of a virtual level at zero are already studied for different self-adjoint operators, see e.g. \cite{Simon1976, Weidl1999}.

In Theorem \ref{t: LT} we obtain estimates of Lieb-Thirring type (see \cite{LiebThirring1976} for the original result and \cite{LaptevWeidl, Frank2016} for reviews of further developments):
\begin{enumerate}
 \item[a)] For most $(\nu, \kappa, \theta)\in\mathfrak M$ we can prove for any $\gamma >0$
\begin{align}\label{LT}
 \tr\big(D_{\theta, \infty}^{\nu, \kappa}(V)\big)_-^\gamma \leqslant K^{\nu, \kappa, \gamma}\int_0^\infty\big\|V_+(r)\big\|_{\mathbb C^{2\times2}}^{1 +\gamma}W^{\nu, \kappa}_{\theta}(r)\,\mathrm dr
\end{align}
with an appropriate weight function $W^{\nu, \kappa}_{\theta} \geqslant 0$ and $K^{\nu, \kappa, \gamma} \in\mathbb R_+$. In many cases we have $W^{\nu, \kappa}_{\theta} \equiv1$.
 \item[b)] In the special case $\beta \in(0, 1/2)$, $\theta =0$ we need to assume $\gamma >2\beta$ and replace \eqref{LT} with
\begin{align}\begin{split}\label{LT modified}
 &\tr\big(D_{0, \infty}^{\nu, \kappa}(V)\big)_-^\gamma \\ &\leqslant K^{\nu, \kappa, \gamma}\bigg(\int_0^\infty\big\|V_+(r)\big\|_{\mathbb C^{2\times2}}^{1 +\gamma -2\beta}r^{-2\beta}\,\mathrm dr +\int_0^\infty\big\|V_+(r)\big\|_{\mathbb C^{2\times2}}^{1 +\gamma}\,\mathrm dr\bigg),
\end{split}
\end{align}
where $K^{\nu, \kappa, \gamma}$ is a finite constant.
\end{enumerate}

The ranges of applicability of the above results are represented in the following table:
\medskip\noindent
\begin{center}
\begin{tabular}{l|ccccc}
 \ &$\!\beta\in \mathrm i\mathbb R_+$ & $\!\beta =0 =\kappa$ & $\!\beta =0 \neq\kappa$ & $\!\beta \in(0, 1/2)$ & $\!\beta \geqslant 1/2$\\ \hline
$\theta =0$ & Ia1 & Ia1 & IIIa & Ib & --- \\
$\theta =\pi/2$ & Ia1 & Ia1 & Ia1 & IIa1 & IIa1 \\
$\theta \in(0, \pi)\setminus\{\frac\pi2\}$ & Ia1 & Ia1 & IIIa & IIa & ---
\end{tabular}\\ \medskip
Table 1
\end{center}

Here the Roman numbers indicate the applicable part of Theorem \ref{t: the alternative}, the letters the part of Theorem \ref{t: LT} and ``$1$'' indicates that \eqref{LT} holds with $W^{\nu, \kappa}_{\theta} \equiv1$.
In the cases marked with ``---'' there exists no self-adjoint realisation, see Theorem \ref{t: self-adjoint realisations}. Note that in the case ``Ia1'' inequality \eqref{LT} is a form of the Hardy-Lieb-Thirring inequality (see \cite{EkholmFrank, Frank2009, FrankLiebSeiringer}).

At last, in Theorem \ref{t: virtual levels 2d} we apply our results to the two-dimensional massless Coulomb-Dirac operator answering the question stated at the beginning of the introduction. A further application to three-dimensional massless Coulomb-Dirac operators can be obtained analogously.

In the following section we explicitly formulate our main results. Their proofs constitute the rest of the article.

\paragraph{Acknowledgement:} S. M. was supported by the RSF grant 15-11-30007.

\section{Main results}

\subsection{Self-adjoint realisations of Coulomb-Dirac operators on the half-line}

First we introduce a family of self-adjoint operators in $\mathsf L^2(\mathbb R_+, \mathbb C^2)$ corresponding to the differential expression \eqref{Dirac symbol}.
We start from the symmetric operator $D^{\nu, \kappa}$ defined on $\mathsf C_0^\infty(\mathbb R_+, \mathbb C^2)$. Then the action of the adjoint operator $(D^{\nu, \kappa})^*$ is also given by \eqref{Dirac symbol}, but on the ``maximal'' domain
\begin{align*}
 \mathfrak D\big((D^{\nu, \kappa})^*\big) =\big\{f\in\mathsf L^2(\mathbb R_+, \mathbb C^2): f \in \mathsf{AC}_{\loc}(\mathbb R_+, \mathbb C^2),\ d^{\nu, \kappa}f\in \mathsf L^2(\mathbb R_+, \mathbb C^2)\big\}.
\end{align*}
In the following theorem we characterise all self--adjoint extensions of $D^{\nu, \kappa}$.
We will make use of the functions $\Psi^{\nu, \kappa}_M$ and $\Psi^{\nu, \kappa}_U:\mathbb R_+\to \mathbb C^2$ introduced in \eqref{Psi_M} and \eqref{Psi_U}.
\begin{thm}\label{t: self-adjoint realisations}\ 
 \begin{enumerate}
  \item For $\beta \in[0, 1/2)$ and $\kappa \neq0$ every self-adjoint extension of $D^{\nu, \kappa}$ coincides with $D^{\nu, \kappa}_\theta$ with some $\theta\in [0,\pi)$, where $D^{\nu, \kappa}_\theta$ is the restriction of $(D^{\nu, \kappa})^*$ to the set of functions $f\in \mathfrak D\big((D^{\nu, \kappa})^*\big)$ satisfying the boundary condition
\begin{align}\label{real bc}
 \lim_{r\to +0}\Big\langle\begin{pmatrix}0 &-1\\ 1& 0\end{pmatrix} f(r), \cos\theta\, \Psi^{\nu, \kappa}_U(r) +\sin\theta\,\Psi^{\nu, \kappa}_M(r)\Big\rangle_{\mathbb C^2}= 0.
\end{align}
  \item For $\beta \geqslant 1/2$ the operator $D^{\nu, \kappa}$ is essentially self-adjoint. We denote its closure by $D^{\nu, \kappa}_{\pi/2}$. Every $f\in\mathfrak D(D^{\nu, \kappa}_{\pi/2})$ satisfies \eqref{real bc} with $\theta :=\pi/2$.
  \item For $\beta\in\mathrm i\mathbb R_+$ or $\kappa =0$ every self-adjoint extension of $D^{\nu, \kappa}$ coincides with $D^{\nu, \kappa}_\theta$ with some $\theta\in [0,\pi)$, where $D^{\nu, \kappa}_\theta$ is the restriction of $(D^{\nu, \kappa})^*$ to the set of functions $f\in \mathfrak D\big((D^{\nu, \kappa})^*\big)$ satisfying the boundary condition
\begin{align}\label{imaginary bc}
 \lim_{r\to +0}\Big\langle\begin{pmatrix}0 &-1\\ 1& 0\end{pmatrix} f(r), \mathrm e^{\mathrm i\theta} \Psi^{\nu, \kappa}_U(r) +\mathrm e^{-\mathrm i\theta}\Psi^{\nu, \kappa}_M(r)\Big\rangle_{\mathbb C^2}= 0.
\end{align}
 \end{enumerate}
\end{thm}

\subsection{Spectral representation of \texorpdfstring{$D^{\nu, \kappa}_\theta$}{D}}

From now on we study the self-adjoint operator $D^{\nu, \kappa}_\theta$ defined in Theorem \ref{t: self-adjoint realisations}.
Using the auxiliary functions and constants introduced in Lemmata \ref{l:c^kappa,nu}, \ref{l:Phi_0}, \ref{l: Wronskian} and \ref{t: spectral function} we are able to find an explicit transform which delivers the spectral representation for $D^{\nu, \kappa}_\theta$:

\begin{thm}\label{t:spectral representation}
Let $\Lambda$ be the operator of multiplication by the independent variable in $\mathsf L^2(\mathbb R, \mathbb C, \mathrm dx)$.
The unitary operator $\mathcal U^{\nu, \kappa}_\theta: \mathsf L^2(\mathbb R_+, \mathbb C^2, \mathrm dr)\to \mathsf L^2(\mathbb R, \mathbb C, \mathrm dx)$ given by
\begin{align}\label{U}
 \mathcal U^{\nu, \kappa}_\theta f :=\underset{R\to\infty}\Ltwolim\sqrt{m^{\nu, \kappa}_\theta(\cdot)}\int_{1/R}^R \big(\Phi_{0, \theta}^{\nu, \kappa}(\cdot; y)\big)^\intercal f(y)\,\mathrm dy
\end{align}
delivers the spectral representation of $D^{\nu, \kappa}_\theta$, i.e. 
\begin{align}\label{spectral representation}
D^{\nu, \kappa}_\theta =(\mathcal U^{\nu, \kappa}_\theta)^*\Lambda\, \mathcal U^{\nu, \kappa}_\theta
\end{align}
holds.
\end{thm}

\subsection{Existence of virtual levels for positively projected\texorpdfstring{\\}{} massless Coulomb-Dirac operators on the half-line}

Let $V$ be a measurable Hermitian $2\times 2$-matrix-function on $\mathbb R_+$. We assume

\vspace{-4mm}
\paragraph{Hypothesis A.} \emph{The operator $P_{\theta, \infty}^{\nu, \kappa}VP_{\theta, \infty}^{\nu, \kappa}$ is relatively form bounded with respect to $D^{\nu, \kappa}_\theta$ with a form bound less than one in $P_{\theta, \infty}^{\nu, \kappa}\mathsf L^2(\mathbb R_+, \mathbb C^2)$.}
\smallskip

Under Hypothesis A the operator \eqref{D positive part with V} is well defined in the form sense and self-adjoint by the KLMN theorem (see e.g. Theorem X.17 in \cite{ReedSimonII}). Hypothesis A is trivially satisfied for $V\in \mathsf L^\infty(\mathbb R_+)$.

The main result concerning the existence of a virtual level is the following theorem.

\begin{thm}\label{t: the alternative}\ 
\begin{enumerate}
 \item[\emph{I.}] Let $A_{\theta}^{\nu, \kappa}:\mathbb R_+\to \mathbb C^2$ be defined by
\begin{align}\label{A_theta}
 A_{\theta}^{\nu, \kappa}(r)\!\!:=\!\! \begin{cases}
                                        \mathrm e^{\mathrm i\theta}r^{\mathrm i\nu}\binom1{\mathrm i} +\mathrm e^{-\mathrm i\theta}r^{-\mathrm i\nu}\binom1{-\mathrm i}, &\text{for }\kappa =0,\ \nu\in\mathbb R,\ \theta\in [0, \pi);\\
					\binom{-\nu}\kappa, &\text{for }\kappa \neq0,\ \beta =0,\ \theta =\pi/2;\\
					\mathrm e^{\mathrm i\theta}r^{-\beta}\kappa\binom{\kappa -\beta}{-\nu} +\mathrm e^{-\mathrm i\theta}r^{\beta}\kappa\binom{\kappa +\beta}{-\nu},\!\! &\text{for }\kappa \neq0,\ \beta \!\in\! \mathrm i\mathbb R_+,\ \theta\!\in\! [0, \pi);\\
					r^{-\beta}\binom\nu{- \kappa -\beta}, &\text{for }\beta \in (0, 1/2),\ \theta =0.\\
                                       \end{cases}
\end{align}
Let $\mathfrak M_{\mathrm I}$ denote the set of all triples $(\nu, \kappa, \theta)$ from the right hand side of \eqref{A_theta}. 
For any $(\nu, \kappa, \theta)\in \mathfrak M_{\mathrm I}$ assume that $V$ satisfies
\begin{align}\begin{split}\label{V integrability}
 \|V\|_{\mathbb C^{2\times2}}\!\in\! \mathsf L^1(\mathbb R_+, r^{2- 2\Re\beta}\mathrm dr)\ \text{and}\ \int_0^\infty \!\!\!\!\big\langle A_{\theta}^{\nu, \kappa}(r), \big|V(r)\big|A_{\theta}^{\nu, \kappa}(r)\big\rangle_{\mathbb C^2}\mathrm dr <\infty.
\end{split}\end{align}
Then the negative spectrum of $D_{\theta, \infty}^{\nu, \kappa}(V)$ is non-empty provided \eqref{VL via A} holds.
 \item[\emph{II.}] Let $\mathfrak M_{\mathrm{II}}$ be the set of $(\nu, \kappa, \theta)\in\mathfrak M$ such that either $\beta> 0$ and $\theta =\pi/2$, or $\beta\in (0, 1/2)$ and $\theta\in (0, \pi)\setminus\{\pi/2\}$ holds. For all $(\nu, \kappa, \theta)\in \mathfrak M_{\mathrm{II}}$ and any $q\in (1, 1 +2\beta)$ there exists $C_q^{\nu, \kappa}> 0$ such that \eqref{integral estimate for beta>0} holds with
\begin{align}\label{CLR weight}
W^{\nu, \kappa}_{\theta, q}(r) :=C_q^{\nu, \kappa}\begin{cases}
                                                   |\cot\theta|^{1 +(q -1)/(2\beta)}r^{-2\beta}, &\text{for }r\leqslant |\cot\theta|^{1/(2\beta)};\\ r^{q -1}, &\text{for }r\geqslant |\cot\theta|^{1/(2\beta)}.
                                                  \end{cases}
\end{align}
Finiteness of the right hand side of \eqref{integral estimate for beta>0} implies Hypothesis A for $V :=V_+$.
 \item[\emph{III.}] Let $V_+\in \mathsf L^\infty(\mathbb R_+, \mathbb C^{2\times2})$. Let $\mathfrak M_{\mathrm{III}}$ be the set of $(\nu, \kappa, \theta)$ such that $\kappa^2 =\nu^2 \neq 0$ and $\theta \in [0, \pi)\setminus\{\pi/2\}$ holds. For all $(\nu, \kappa, \theta)\in \mathfrak M_{\mathrm{III}}$ there exists a finite constant $K^{\nu, \kappa}$ independent of $\theta$ such that the estimate \eqref{no VL for beta=0} holds.
\end{enumerate}
\end{thm}

\subsection{Lieb-Thirring type estimates on the negative eigenvalues of \texorpdfstring{$D^{\nu, \kappa}_{\theta, \infty}(V)$}{D(V)}}

In Theorem \ref{t: LT} we provide estimates on the sums of powers of negative eigenvalues of $D^{\nu, \kappa}_{\theta, \infty}(V)$ via weighted integrals of powers of the perturbation potential $V$. In the case of $\theta =\pi/2$ the result is fully analogous to the classical Lieb-Thirring estimate.

\begin{thm}\label{t: LT}\ 
\begin{enumerate}
 \item[a)] For $\nu, \kappa\in \mathbb R$, $\theta \in[0, \pi)$ % The following estimates hold true with appropriate finite constants $k^{\nu, \kappa}$:
let $W^{\nu, \kappa}_{\theta}: \mathbb R_+\to \mathbb R_+$ be defined by
\begin{align}\label{weights}
W^{\nu, \kappa}_{\theta}(r)\!\!:=\!\!\begin{cases}
                    1, &\hspace{-3mm}\begin{cases}\text{for $\nu^2 = \kappa^2\neq 0$ and $\theta =\pi/2$};\\ \text{for $\beta \in \mathrm i\mathbb R_+$ and $\theta \in[0, \pi)$};\\ \text{for $\kappa =0$, $\nu\in\mathbb R$ and $\theta\in[0, \pi)$};\end{cases}\\
		    \max\big\{\!\!-\ln(\mathrm e^{\tan\theta}r), 1\big\}^2\!\!,\!\!\!\!&\text{for $\nu^2\! = \kappa^2\! \neq 0$ and $\theta\! \in\! [0, \pi)\!\setminus\!\{\pi/2\}$};\\
		    \max\big\{1, |\cot\theta|r^{-2\beta}\big\}, &\hspace{-3mm}\begin{cases}\text{for $\beta \in (0, 1/2)$ and $\theta \in (0, \pi)$};\\ \text{for $\beta \geqslant 1/2$ and $\theta =\pi/2$}.\end{cases}
                   \end{cases}
\end{align}
Then for any $\gamma >0$ there exists $K^{\nu, \kappa, \gamma} >0$ such that \eqref{LT} holds.
\item[b)] For $\beta\in(0, 1/2)$ and any $\gamma >2\beta$ there exists $K^{\nu, \kappa, \gamma} >0$ such that \eqref{LT modified} holds.
\end{enumerate}
In both cases the finiteness of the right hand sides of \eqref{LT} or \eqref{LT modified} implies that $V :=V_+$ satisfies Hypothesis A.
\end{thm}

\subsection{Application to two-dimensional projected Coulomb-Dirac operators}

Let us now consider positively projected massless Coulomb-Dirac operators in two dimensions. Such operators are relevant for description of graphene with Coulomb impurity \cite{MorozovMueller}.
Every $u\in \mathsf L^2(\mathbb R^2)$ can be represented in the polar coordinates as
\begin{equation*}%\label{polar decomposition}
 u(r, \varphi)= \frac1{\sqrt{2\pi}}\sum_{m\in \mathbb Z}r^{-1/2}u_m(r)\mathrm e^{\mathrm im\varphi}
\end{equation*}
with
\begin{equation*}%\label{polar components}
 u_m(r):= \sqrt{\frac r{2\pi}}\int_0^{2\pi}u(r, \varphi)\mathrm e^{-\mathrm im\varphi}\mathrm d\varphi.
\end{equation*}
Introducing the unitary angular momentum decomposition
\begin{equation}\label{A}
 \mathcal A:\mathsf L^2(\mathbb R^2, \mathbb C^2)\to \underset{\varkappa\in \mathbb Z+1/2}\bigoplus\mathsf L^2(\mathbb R_+, \mathbb C^2), \quad \binom{v}{w}\mapsto \underset{\varkappa\in \mathbb Z+1/2}\bigoplus\binom{v_{\varkappa-1/2}}{\mathrm iw_{\varkappa+1/2}}
\end{equation}
we observe (see e.g. \cite{MorozovMueller}) that for any self-adjoint realisation of $-\mathrm i\boldsymbol\sigma\cdot\nabla -\nu|\cdot|^{-1}$ in $\mathsf L^2(\mathbb R^2, \mathbb C^2)$ there exists a map
\begin{align}\label{theta}
\boldsymbol\theta: \mathbb Z+1/2 \to [0,\pi) \quad\text{with} \quad \boldsymbol\theta(\kappa) =\pi/2 \quad \text{for all} \quad \kappa^2 \geqslant\nu^2 +1/4
\end{align}
such that the self-adjoint operator in question coincides with
\begin{align}\label{Coulomb-Dirac 2D}
 D^{\nu}_{\boldsymbol\theta} :=\mathcal A^*\bigg(\bigoplus_{\kappa\in \mathbb Z+1/2}D^{\nu, \kappa}_{\boldsymbol\theta(\kappa)}\bigg)\mathcal A,
\end{align}
where the components on the right hand side are defined in Theorem \ref{t: self-adjoint realisations}. On the other hand, every $\boldsymbol\theta$ satisfying \eqref{theta} gives rise to a self-adjoint realisation of $-\mathrm i\boldsymbol\sigma\cdot\nabla -\nu|\cdot|^{-1}$ in $\mathsf L^2(\mathbb R^2, \mathbb C^2)$ via \eqref{Coulomb-Dirac 2D}.

Let $P^{\nu}_{\boldsymbol\theta} :=\mathcal A^*(\bigoplus_{\kappa\in \mathbb Z +1/2}P_{\theta, \infty}^{\nu, \kappa})\mathcal A$ be the spectral projector onto the positive spectral subspace of $D^{\nu}_{\boldsymbol\theta}$.
For measurable Hermitian $(2\times2)$-matrix functions $Q$ on $\mathbb R^2$ we are interested in the negative spectrum of
\begin{align}\label{D projected with V 2d}
 D^{\nu}_{\boldsymbol\theta}(Q) :=P^{\nu}_{\boldsymbol\theta}(D^{\nu}_{\boldsymbol\theta} -Q)P^{\nu}_{\boldsymbol\theta}
\end{align}
in $P^{\nu}_{\boldsymbol\theta}\mathsf L^2(\mathbb R^2, \mathbb C^2)$.
We assume

\vspace{-4mm}
\paragraph{Hypothesis B.} \emph{The operator $P^{\nu}_{\boldsymbol\theta}QP^{\nu}_{\boldsymbol\theta}$ is relatively form bounded with respect to $D^{\nu}_{\boldsymbol\theta}$ with a relative bound less than one in $P^{\nu}_{\boldsymbol\theta}\mathsf L^2(\mathbb R_+, \mathbb C^2)$.}
\smallskip

Under Hypothesis B the operator \eqref{D projected with V 2d} is well-defined via its quadratic form and self-adjoint by KLMN theorem. Hypothesis B is fulfilled for all $Q\in\mathsf L^\infty(\mathbb R^2)$.

\begin{thm}\label{t: virtual levels 2d}Let $\nu \in\mathbb R$. 
\begin{enumerate}
 \item Suppose that there exists $\kappa_0\in \mathbb Z+1/2$ such that $\big(\nu, \kappa_0, \boldsymbol\theta(\kappa_0)\big)\in \mathfrak M_{\mathrm{I}}$ as defined in Theorem \ref{t: the alternative}. Then $D^{\nu}_{\boldsymbol\theta}(Q)$ has non-empty negative spectrum provided 
\begin{align}\label{V}
 V := \frac1{2\pi}\int_0^{2\pi}\begin{pmatrix}
                   Q_{11}(\cdot, \varphi) & -\mathrm i Q_{12}(\cdot, \varphi)\mathrm e^{\mathrm i\varphi}\\ \mathrm i Q_{21}(\cdot, \varphi)\mathrm e^{-\mathrm i\varphi} & Q_{22}(\cdot, \varphi)
                  \end{pmatrix}\mathrm d\varphi
\end{align}
satisfies \eqref{V integrability} and \eqref{VL via A} with $(\nu, \kappa, \theta):= \big(\nu, \kappa_0, \boldsymbol\theta(\kappa_0)\big)$.
 \item Suppose that for all $\kappa\in \mathbb Z+1/2$ the triple $\big(\nu, \kappa, \boldsymbol\theta(\kappa)\big)$ does not belong to $\mathfrak M_{\mathrm{I}}$. If there exists a measurable function $R: \mathbb R_+\to \overline{\mathbb R_+}$ with
\begin{align}\label{R bounds Q}
 Q \leqslant R\big(|\cdot|\big)\mathbb I \quad\text{almost everywhere in }\mathbb R^2
\end{align}
(where $\mathbb I$ denotes the $2\times 2$ identity matrix) such that:
\begin{enumerate}
 \item For all $\kappa\in \mathbb Z+1/2$ with $\kappa^2\leqslant \nu^2 +1/4$ such that $\big(\nu, \kappa, \boldsymbol\theta(\kappa)\big)\in \mathfrak M_{\mathrm{II}}$ there exists $q\in (1, 1 +2\beta)$ such that the right hand side of \eqref{integral estimate for beta>0} with $V :=R\mathbb I$ and $W^{\nu, \kappa}_{\theta, q}$ defined in \eqref{CLR weight} is finite;
 \item For all $\kappa\in \mathbb Z+1/2$ with $\kappa^2\leqslant \nu^2 +1/4$ such that $\big(\nu, \kappa, \boldsymbol\theta(\kappa)\big)\in \mathfrak M_{\mathrm{III}}$ we have $R\in \mathsf L^\infty(\mathbb R_+)$ and the right hand side of \eqref{no VL for beta=0} with $V :=R\mathbb I$ is finite;
 \item $R \in \mathsf L^\infty(\mathbb R_+, r\mathrm dr) +\mathsf L^2(\mathbb R_+, r\mathrm dr)$
\end{enumerate}
then there exists $\alpha_c >0$ such that for all $\alpha\in [0, \alpha_c)$ operator $D^{\nu}_{\boldsymbol\theta}(\alpha Q)$ has no negative spectrum. Note that for any $Q\in \mathsf C_0^\infty\big(\mathbb R^2\big)$ all the assumptions are satisfied.
\end{enumerate}
\end{thm}

For $\nu \in [0, 1/2]$ there exists a distinguished self-adjoint realisation $D^\nu$ of $-\mathrm i\boldsymbol\sigma\cdot\nabla -\nu|\cdot|^{-1}$, which coincides with $D^{\nu}_{\boldsymbol\theta_0}$ with $\boldsymbol\theta_0(\kappa) :=\pi/2$ for all $\kappa\in\mathbb Z+ 1/2$, see \cite{Mueller, MorozovMueller}. We write $D^\nu(Q)$ for $D^{\nu}_{\boldsymbol\theta_0}(Q)$. Theorem \ref{t: virtual levels 2d} implies
\begin{cor}
Suppose that $V$ defined in \eqref{V} satisfies 
\begin{align*}
 \|V\|_{\mathbb C^{2\times2}}\in \mathsf L^1\big(\mathbb R_+, (1 +r^{2})\mathrm dr\big)\ \text{and}\ \int_0^\infty\bigg\langle\binom{-1}1, V(r)\binom{-1}1\bigg\rangle_{\mathbb C^{2}}\mathrm dr >0.
\end{align*}
Then for any $\alpha >0$ the negative spectrum of $D^{1/2}(\alpha Q)$ is non-empty.
\end{cor}

For $\nu \in[0, 1/2)$ an application of Theorem \ref{t: virtual levels 2d}, part 2 to $D^{\nu}(Q)$ gives a weaker result than that of Theorem 3 in \cite{MorozovMueller}, which gives the following bound on the amount of negative eigenvalues of $D^{\nu}(Q)$: There exists $C^{\mathrm{CLR}}_\nu> 0$ such that
\begin{equation}\label{CLR}
 \rank \big(D^\nu(Q)\big)_- \leqslant C^{\mathrm{CLR}}_\nu\int_{\mathbb R^2} \tr \big(Q_+(\mathbf x)\big)^2\dd\mathbf x.
\end{equation}
Indeed, even for $Q :=R\big(|\cdot|\big)\mathbb I$ with radial non-negative $R$ the hypothesis of Theorem \ref{t: virtual levels 2d} requires $R \in \mathsf L^\infty(\mathbb R_+, r\mathrm dr) +\mathsf L^2(\mathbb R_+, r\mathrm dr)$ and
\begin{align*}
 \int_0^\infty R^q(r)r^{q -1}\mathrm dr <\infty \quad \text{for some }q \in (1, 1 +2\sqrt{1/4 -\nu^2}).
\end{align*}
But then we can write $R =R_1 + R_2$ with $R_1\in \mathsf L^\infty(\mathbb R_+, r\mathrm dr)$, $R_2\in \mathsf L^2(\mathbb R_+, r\mathrm dr)$ and $R_1, R_2\geqslant 0$. Hence we have
\begin{align*}\begin{split}
 &\int_{\mathbb R^2} \tr \big(Q_+(\mathbf x)\big)^2\dd\mathbf x =4\int_0^\infty \big(R_1(r) +R_2(r)\big)^2r\mathrm dr\\ &\leqslant 8\big\|R_1\big\|^{2 -q}_{\mathsf L^\infty(\mathbb R_+, r\mathrm dr)}\int_0^\infty R_1^q(r)r^{q -1}\mathrm dr + 8\int_0^\infty R_2^2(r)r\mathrm dr<\infty,
\end{split}\end{align*}
so \eqref{CLR} implies the statement of Theorem \ref{t: virtual levels 2d}, part 2 under weaker assumptions.

\section{Proof of Theorem \ref{t: self-adjoint realisations}}

The fundamental solution of $d^{\nu, \kappa}\Psi^{\nu, \kappa} =0$
is a linear combination of $\Psi^{\nu, \kappa}_M$ and $\Psi^{\nu, \kappa}_U$ with
\begin{align}\label{Psi_M}
 \Psi^{\nu, \kappa}_M(r) := \begin{cases}\kappa r^\beta\binom{\kappa +\beta}{-\nu},& \text{for }\nu^2\neq \kappa^2 \neq0\text{ and }\kappa \neq -\beta;\\ \beta r^{\beta}\binom01,& \text{for }\beta =-\kappa \in\mathbb R_+;\\ \binom{-\nu}\kappa, & \text{for }\nu^2 =\kappa^2 \neq0;\\ r^{-\mathrm i\nu} \binom1{-\mathrm i}, & \text{for }\kappa =0\end{cases}
 \end{align}
and
\begin{align}\label{Psi_U}
 \Psi^{\nu, \kappa}_U(r) := \begin{cases}
   r^{-\beta}\binom{\nu}{-\kappa -\beta},& \text{for }\beta\in \mathbb R_+\text{ and }\kappa \neq -\beta;\\
   r^{-\beta}\binom10,& \text{for }\beta =-\kappa \in\mathbb R_+;\\
   \kappa r^{-\beta}\binom{\kappa -\beta}{-\nu},& \text{for }\beta\in \mathrm i\mathbb R_+ \text{ and }\kappa \neq0;\\
   \ln r\binom{-\nu}\kappa -\frac1{2\kappa}\binom{\nu}{\kappa},\ & \text{for }\nu^2 =\kappa^2 \neq0;\\
   r^{\mathrm i\nu} \binom1{\mathrm i}, & \text{for }\kappa =0.
 \end{cases}
\end{align}

Since none of the solutions belongs to $\mathsf L^2\big((1, \infty)\big)$, \eqref{Dirac symbol} is always in the limit point case at infinity. For $\beta\geqslant 1/2$, $\Psi^{\nu, \kappa}_U$ does not belong to $\mathsf L^2\big((0, 1)\big)$ and we have a limit point case at zero. Otherwise there is a limit circle case at zero. The rest follows from Theorem 1.5 in \cite{Weidmann1971}.

\section{Proof of Theorem \ref{t:spectral representation}}

We begin by studying the classical solutions of the spectral equation
\begin{equation}\label{spectral equation}
 d^{\nu, \kappa}f =\lambda f
% =\begin{pmatrix}
%   -\nu/r & -\frac{\mathrm d}{\mathrm dr}- \frac{\kappa}{r}\\ \frac{\mathrm d}{\mathrm dr}- \frac{\kappa}{r} & -\nu/r
%  \end{pmatrix}\binom {f_1(r)}{f_2(r)} =\lambda\binom {f_1(r)}{f_2(r)},
\end{equation}
on the half-line $\mathbb R_+$, where $\lambda\in \mathbb C\setminus \mathrm i\overline{\mathbb R_+}$ is the spectral parameter.

\begin{lem}\label{Lemma solutions dirac equation}
Let $M$ and $U$ be the Kummer functions (see \cite{dlmf}, Section 13.2). For $\lambda =1$ any classical solution of \eqref{spectral equation} is a linear combination of the following two linearly independent functions on $\mathbb R_+$:
\begin{enumerate}
 \item For $\nu^2\neq \kappa^2 \neq0$ and $\beta \neq -\kappa$,
\begin{align*}
 \begin{split}\Phi^{\nu, \kappa}_M(r) &:=r^{\beta}\mathrm e^{-\mathrm ir}\Bigg(\beta(\kappa +\beta +\mathrm i\nu)M(\mathrm i\nu +\beta, 2\beta, 2\mathrm ir)\binom1{-\mathrm i}\\ &+(\nu -\mathrm i\beta)M(1 +\mathrm i\nu +\beta, 1 +2\beta, 2\mathrm ir)\binom\nu{-\kappa -\beta}\Bigg)\end{split}
\end{align*}
and
\begin{align}\label{Phi_U tilde beta neq 0}
 \begin{split}\widetilde\Phi^{\nu, \kappa}_U(r) &:=r^{\beta}\mathrm e^{-\mathrm ir}\Bigg((\kappa +\beta +\mathrm i\nu)U(\mathrm i\nu +\beta, 2\beta, 2\mathrm ir)\binom1{-\mathrm i}\\ &+2(\mathrm i\beta -\nu)U(1 +\mathrm i\nu +\beta, 1 +2\beta, 2\mathrm ir)\binom\nu{-\kappa -\beta}\Bigg);\end{split}
\end{align}
\item for $\beta =-\kappa \in\mathbb R_+$,
\begin{align*}
 \begin{split}\Phi^{\nu, \kappa}_M(r) &:=\beta r^{\beta}\mathrm e^{-\mathrm ir}\Bigg(M(\beta, 2\beta, 2\mathrm ir)\binom{\mathrm i}1 +M(1 +\beta, 1 +2\beta, 2\mathrm ir)\binom{-\mathrm i}0\Bigg),\end{split}
\end{align*}
and
\begin{align}\label{Phi_U tilde beta neq -kappa}
 \begin{split}\widetilde\Phi^{\nu, \kappa}_U(r) &:=r^{\beta}\mathrm e^{-\mathrm ir}\Bigg(U(\beta, 2\beta, 2\mathrm ir)\binom{\mathrm i}1 +U(1 +\beta, 1 +2\beta, 2\mathrm ir)\binom{2\mathrm i\beta}0\Bigg)\end{split}
\end{align}
 \item for $\nu^2= \kappa^2 \neq0$,
\begin{align*}
 \begin{split}\Phi^{\nu, \kappa}_M(r) :=\mathrm e^{-\mathrm ir}\Bigg(M(1 +\mathrm i\nu, 2, 2\mathrm ir)&\bigg((\kappa +\mathrm i\nu)r\binom1{-\mathrm i} -\binom{\nu}{-\kappa}\bigg)\\ &+(\nu -\mathrm i)rM(2 +\mathrm i\nu, 3, 2\mathrm ir)\binom\nu{-\kappa}\Bigg)\end{split}
\end{align*}
and
\begin{align}
 \begin{split}\label{Phi_U tilde beta=0}\widetilde\Phi^{\nu, \kappa}_U(r) :=\mathrm e^{-\mathrm ir}\Bigg(U(1 +\mathrm i\nu, 2, 2\mathrm ir)&\bigg((\kappa +\mathrm i\nu)r\binom1{-\mathrm i} -\binom{\nu}{-\kappa}\bigg)\\ &+2(\mathrm i -\nu)rU(2 +\mathrm i\nu, 3, 2\mathrm ir)\binom\nu{-\kappa}\Bigg);\end{split}
\end{align}
 \item for $\kappa =0$, 
\begin{align}\label{Phi for kappa=0}
 \Phi^{\nu, 0}_M(r) :=r^{-\mathrm i\nu}\mathrm e^{-\mathrm ir}\binom1{-\mathrm i}, \qquad \Phi^{\nu, 0}_U(r) :=r^{\mathrm i\nu}\mathrm e^{\mathrm ir}\binom1{\mathrm i}.
\end{align}

\end{enumerate}
It is convenient to replace $\widetilde\Phi^{\nu, \kappa}_U$ by another solution:\\
For $0 <\beta <1/2$,
\begin{align}\label{Phi_U beta 0..1/2}
 \Phi^{\nu, \kappa}_U(r) := 
\frac{-2^{2\beta -1}\mathrm i\mathrm e^{\mathrm i\pi\beta}\Gamma(\beta +\mathrm i\nu)}{\Gamma(2\beta)}\Big(\widetilde\Phi^{\nu, \kappa}_U(r) +\frac{2\Gamma(-2\beta)}{\Gamma(1 -\beta +\mathrm i\nu)}\Phi^{\nu, \kappa}_M(r)\Big);
\end{align}
for $\beta \geqslant1/2$,
\begin{align*}
 \Phi^{\nu, \kappa}_U(r) := 
\frac{-2^{2\beta -1}\mathrm i\mathrm e^{\mathrm i\pi\beta}\Gamma(\beta +\mathrm i\nu)}{\Gamma(2\beta)}\widetilde\Phi^{\nu, \kappa}_U(r);
\end{align*}
for $\beta \in \mathrm{i}\mathbb R_+$ and $\kappa \neq 0$,
\begin{align}
 \begin{split}\label{Phi_U beta imaginary}
\Phi^{\nu, \kappa}_U(r) &:= \frac{-2^{2\beta -1}\mathrm i\kappa(\kappa -\beta)\mathrm e^{\mathrm i\pi\beta}\Gamma(\beta +\mathrm i\nu)}{\nu\Gamma(2\beta)}\Big(\widetilde\Phi^{\nu, \kappa}_U(r) +\frac{2\Gamma(-2\beta)}{\Gamma(1 -\beta +\mathrm i\nu)}\Phi^{\nu, \kappa}_M(r)\Big);
 \end{split}
\end{align}
and for $\nu^2= \kappa^2 \neq0$,
\begin{align}\label{Phi_U beta=0}
 \Phi^{\nu, \kappa}_U(r) := \Gamma(\mathrm i\nu)\widetilde\Phi^{\nu, \kappa}_U(r) -\Big(\ln2  +2\gamma +\psi(1+\mathrm{i}\nu) +\frac{\mathrm{i}\pi}{2} +\frac{\mathrm{i}}{2 \nu}\Big)\Phi^{\nu, \kappa}_M(r),
\end{align}
where $\psi$ is the digamma function and $\gamma$ is the Euler-Mascheroni constant.
The above solutions allow unique analytic continuations from $\mathbb R_+$ to $\mathbb C\setminus \mathrm i\overline{\mathbb R_+}$.
For arbitrary $\lambda\in \mathbb C\setminus \mathrm i\overline{\mathbb R_+}$ any classical solution to \eqref{spectral equation} is given by a linear combination of $\Phi^{\nu, \kappa}_M(\lambda\,\cdot)$ and $\Phi^{\nu, \kappa}_U(\lambda\,\cdot)$.
\end{lem}

\begin{proof}
The validity of Lemma \ref{Lemma solutions dirac equation} can be checked by a straightforward calculation using functional relations between Kummer functions and their derivatives, see e.g. \cite{dlmf}, 13.3.13--15, 13.3.22. It is, however, more instructive to provide a derivation of the solutions, which we will now do for $\kappa \neq 0$. \\
We start by seeking the solutions to \eqref{spectral equation} with $\lambda =1$ in the form
\begin{align}\label{FG}
f(r) =r^\beta F(r)
\end{align}
obtaining
\begin{equation}\label{Dirac equation FG}
 \begin{pmatrix}
  -\nu/r -1& -\frac{\mathrm d}{\mathrm dr}- \frac{\beta +\kappa}{r}\\ \frac{\mathrm d}{\mathrm dr}+ \frac{\beta -\kappa}{r} & -\nu/r -1
 \end{pmatrix}\binom {F_1(r)}{F_2(r)} =0.
\end{equation}
Introducing
\begin{align}
 \binom {G_1}{G_2} :=\begin{cases}\begin{pmatrix}\frac{\kappa +\beta}2F_1+ \frac\nu2F_2\\ \frac\nu2F_1 -\frac{\kappa +\beta}2F_2\end{pmatrix},& \text{ for }\beta\neq \kappa \neq0,\\ \displaystyle\binom {F_2}{F_1},& \text{ for }\beta= -\kappa \neq0\end{cases}\label{HL}
\end{align}
we obtain that \eqref{Dirac equation FG} is equivalent to
\begin{align}
 G_2 &=-G_1',\label{L}\\ G_1''(r) +\frac{2\beta}rG_1'(r) +\Big(1 +\frac{2\nu}r\Big)G_1(r) &=0.\label{H equation}
\end{align}
We can now find two linearly independent solutions of \eqref{H equation} analytic in $\mathbb C\setminus\mathrm i\overline{\mathbb R_+}$.

For $\beta \neq0$ making the substitution
\begin{align}
 G_1(r) =\mathrm e^{-\mathrm ir}w(2\mathrm ir), \qquad z:= 2\mathrm ir\label{w}
\end{align}
we obtain
\begin{align*}
 zw''(z) +(2\beta -z)w'(z) -(\mathrm i\nu +\beta)w(z) =0, \quad z\in\mathbb C\setminus\overline{\mathbb R_-},
\end{align*}
which is the Kummer equation with parameters $a :=\mathrm i\nu +\beta$, $b :=2\beta$ (see e.g. \cite{dlmf}, 13.2.1).

For $\beta =0$, the substitution
\begin{align}\label{G_1 substitution}
 G_1(r) =r\mathrm e^{-\mathrm ir}v(2\mathrm ir), \qquad z:= 2\mathrm ir
\end{align}
delivers
\begin{align*}
 zv''(z) +(2 -z)v'(z) -(1 +\mathrm i\nu)v(z) =0, \quad z\in\mathbb C\setminus \overline{\mathbb R_-},%\label{v}
\end{align*}
i.e. the Kummer equation with parameters $a :=1 +\mathrm i\nu$, $b :=2$.

In both cases the linearly independent solutions of the Kummer equation can be chosen as the Kummer functions $M(a, b, z)$ and $U(a, b, z)$ (see \cite{dlmf}, Section 13.2).
Substituting back into \eqref{w} (or \eqref{G_1 substitution}), \eqref{L} (together with 13.3.15 and 13.3.22 in \cite{dlmf}), \eqref{HL} and \eqref{FG}, we obtain the two independent solutions of \eqref{spectral equation} with $\lambda =1$ analytic in $\mathbb C\setminus \mathrm i\overline{\mathbb R_+}$ stated in the lemma.
\end{proof}

\begin{lem}\label{lemma asymptotics at zero}
 The solutions $\Phi^{\nu, \kappa}_M$ and $\Phi^{\nu, \kappa}_U$ have the following asymptotics at zero:
\begin{align}\label{Phi_M at zero}
 \Phi^{\nu, \kappa}_M(r) = \begin{cases}\kappa r^\beta\binom{\kappa +\beta}{-\nu} +O(r^{1 +\beta}),&\text{for }\nu^2\neq \kappa^2 \neq0\text{ and }\kappa \neq -\beta;\\ \beta r^{\beta}\binom01 +O(r^{1 +\beta}),& \text{for }\beta =-\kappa \in\mathbb R_+;\\ \binom{-\nu}\kappa +O(r), & \text{for }\nu^2 =\kappa^2 \neq0;\\ r^{-\mathrm i\nu} \binom1{-\mathrm i} +O(r), & \text{for }\kappa =0;\end{cases}
 \end{align}
\begin{align}\label{Phi_U at zero}
 \Phi^{\nu, \kappa}_U(r) = \begin{cases}
   r^{-\beta}\binom{\nu}{-\kappa -\beta} +O(r^{1 -\beta}),& \text{for }\beta\in \mathbb R_+\setminus\{1/2, -\kappa\};\\
   r^{-1/2}\binom\nu{-\kappa -1/2} +O(r^{1/2}\ln r),& \text{for }\beta =1/2 \neq-\kappa;\\
   r^{-\beta}\binom10 +O(r^{1 -\beta}),& \text{for }\beta =-\kappa \in\mathbb R_+\setminus\{1/2\};\\
   r^{-1/2}\binom10 +O(r^{1/2}\ln r),& \text{for }\beta =-\kappa =1/2;\\
   \kappa r^{-\beta}\binom{\kappa -\beta}{-\nu} +O(r)  ,&\text{for }\beta\in \mathrm i\mathbb R_+ \text{ and }\kappa \neq0;\\
   \ln r\binom{-\nu}\kappa -\frac1{2\kappa}\binom{\nu}{\kappa}+O(r\ln r),\ & \text{for }\nu^2 =\kappa^2 \neq0;\\
   r^{\mathrm i\nu} \binom1{\mathrm i} +O(r), & \text{for }\kappa =0.
 \end{cases}
\end{align}
For $\beta\in [0, 1/2)$ and $\kappa \neq 0$ both $\Phi^{\nu, \kappa}_M$ and $\Phi^{\nu, \kappa}_U$ are real-valued, for $\beta\in \mathrm i\mathbb R_+$ or $\kappa=0$ they are complex conjugate of each other.
\end{lem}

\begin{proof}
The asymptotics of $\Phi^{\nu, \kappa}_M$ follows from the definitions of Lemma \ref{Lemma solutions dirac equation} and 13.2.2 of \cite{dlmf}.

For $\beta\in \mathbb R_+\setminus(\mathbb N/2)$ or $\beta\in \mathrm i\mathbb R_+$ and $\kappa \neq 0$,
the expansion of $\widetilde\Phi^{\nu, \kappa}_U$ follows from \eqref{Phi_U tilde beta neq 0} or \eqref{Phi_U tilde beta neq -kappa} and 13.2.42 in \cite{dlmf}. For $\beta\in (0, 1/2)$ and $\beta\in \mathrm i\mathbb R_+$, the definitions \eqref{Phi_U beta 0..1/2} and \eqref{Phi_U beta imaginary} are constructed in such a way that the coefficients at $r^\beta$ in the asymptotics \eqref{Phi_U at zero} are zero.

For $\beta \in \mathbb N/2$ or $\beta =0$ and $\kappa \neq0$ the expansion of $\widetilde\Phi^{\nu, \kappa}_U$ follows from \eqref{Phi_U tilde beta neq 0} or \eqref{Phi_U tilde beta=0} and 13.2.9 in \cite{dlmf}. In \eqref{Phi_U beta=0} the linear combination of $\widetilde\Phi^{\nu, \kappa}_U$ and $\Phi^{\nu, \kappa}_M$ is chosen in such a way that \eqref{Phi_U at zero} holds true.

Since all the entries of \eqref{Dirac symbol} are invariant under complex conjugation, for $\lambda\in \mathbb R$ the real and imaginary parts of any solution to \eqref{spectral equation} are again solutions to \eqref{spectral equation}.

For $\beta\in [0, 1/2)$ and $\kappa \neq 0$, the imaginary parts of $\Phi^{\nu, \kappa}_M$ and $\Phi^{\nu, \kappa}_U$ are solutions of \eqref{spectral equation} with $\lambda =1$, thus linear combinations of $\Phi^{\nu, \kappa}_M$ and $\Phi^{\nu, \kappa}_U$. On the other hand, by \eqref{Phi_M at zero} and \eqref{Phi_U at zero} this imaginary parts are $O(r^{1 -\beta})$. Hence they must vanish identically.
Analogously, for $\beta\in \mathrm i\mathbb R_+$ or $\kappa=0$, the imaginary part of $\Phi^{\nu, \kappa}_M +\Phi^{\nu, \kappa}_U$ and the real part of $\Phi^{\nu, \kappa}_M -\Phi^{\nu, \kappa}_U$ are solutions of \eqref{spectral equation} with $\lambda =1$, thus linear combinations of $\Phi^{\nu, \kappa}_M$ and $\Phi^{\nu, \kappa}_U$. Since this combinations are $O(r)$ by \eqref{Phi_M at zero} and \eqref{Phi_U at zero}, they must be identically zero. Thus $\Phi^{\nu, \kappa}_U =\overline{\Phi^{\nu, \kappa}_M}$ holds.
\end{proof}

\begin{lem}\label{l:c^kappa,nu}
For $\kappa \neq0$, $\nu\in\mathbb R$ and $\lambda \in\mathbb C\setminus(\mathbb R\cup\mathrm i\mathbb R_+)$ let
\begin{align}\label{c^kappa,nu}
  c^{\nu, \kappa}(\lambda):= \begin{cases}
                              c^{\nu, \kappa}_{+,\pm}, &\text{for }\pm\Re\lambda >0,\, \Im\lambda >0;\\
			      c^{\nu, \kappa}_{-}, &\text{for }\Im\lambda <0
                             \end{cases}
\end{align}
with
\begin{align}\label{c_+,pm}
 c^{\nu, \kappa}_{+,\pm}\!\! :=\!\!\begin{cases}
                            \dfrac{\mathrm i2^{2\beta -1}\big|\Gamma(\beta +\mathrm i\nu)\big|^2\mathrm e^{\pm\pi\nu}\mathrm e^{(1 \mp1)\mathrm i\pi\beta}}{\beta\Gamma^2(2\beta)} +\dfrac{\mathrm i2^{2\beta}\Gamma(\beta +\mathrm i\nu)\mathrm e^{\mathrm i\pi\beta}\Gamma(-2\beta)}{\Gamma(2\beta)\Gamma(1 -\beta +\mathrm i\nu)},&\\ \hfill \text{for }\beta \in(0, 1/2);&\\
			    \dfrac{\mathrm i2^{2\beta -1}\big|\Gamma(\beta +\mathrm i\nu)\big|^2\mathrm e^{\pm\pi\nu}\mathrm e^{(1 \mp1)\mathrm i\pi\beta}}{\beta\Gamma^2(2\beta)}, \hfill \text{for }\beta \geqslant 1/2;&\\
			    \dfrac{\mathrm i\kappa(\kappa -\beta)2^{2\beta -1}\Gamma(\beta +\mathrm i\nu)\mathrm e^{\mathrm i\pi\beta}}{\nu\Gamma(2\beta)}\Big(\dfrac{\Gamma(\beta -\mathrm i\nu)\mathrm e^{\pm\pi\nu \mp\mathrm i\pi\beta}}{\beta\Gamma(2\beta)} +\dfrac{2\Gamma(-2\beta)}{\Gamma(1 -\beta +\mathrm i\nu)}\Big),&\\ \hfill\text{ for }\beta \in\mathrm i\mathbb R_+;&\\
			    \Gamma(1 -\mathrm i\nu)\Gamma(\mathrm i\nu)\mathrm e^{\pm\pi\nu} +\ln2  +2\gamma +\psi(1+\mathrm{i}\nu) +\dfrac{\mathrm{i}\pi}{2} +\dfrac{\mathrm{i}}{2 \nu},&\\ \hfill \text{for }\beta =0&
                           \end{cases}
\end{align}
and
\begin{align}\label{c_-}
 c^{\nu, \kappa}_{-} :=\begin{cases}
                            \dfrac{\mathrm i2^{2\beta}\Gamma(\beta +\mathrm i\nu)\mathrm e^{\mathrm i\pi\beta}\Gamma(-2\beta)}{\Gamma(2\beta)\Gamma(1 -\beta +\mathrm i\nu)}, & \text{for }\beta \in(0, 1/2);\\
			    0, & \text{for }\beta \geqslant 1/2;\\
			    \dfrac{\mathrm i\kappa(\kappa -\beta)2^{2\beta}\Gamma(\beta +\mathrm i\nu)\Gamma(-2\beta)\mathrm e^{\mathrm i\pi\beta}}{\nu\Gamma(2\beta)\Gamma(1 -\beta +\mathrm i\nu)}, & \text{for }\beta \in\mathrm i\mathbb R_+;\\
			    \ln2  +2\gamma +\psi(1+\mathrm{i}\nu) +\dfrac{\mathrm{i}\pi}{2} +\dfrac{\mathrm{i}}{2 \nu}, & \text{for }\beta =0.
                           \end{cases}
\end{align}
Then
\begin{align}\label{Phi_infty}
 \Phi_\infty^{\nu, \kappa}(\lambda; r):= \begin{cases}
                                          \Phi_U^{\nu, \kappa}(\lambda r) +c^{\nu, \kappa}(\lambda)\Phi_M^{\nu, \kappa}(\lambda r),& \text{ for }\kappa \neq0,\ \lambda\in \mathbb C\setminus (\mathbb R\cup\mathrm i\mathbb R_+);\\
					  \Phi_U^{\nu, 0}(\lambda r),& \text{ for }\kappa =0,\ \Im\lambda> 0,\ \Re \lambda\neq 0;\\
					  \Phi_M^{\nu, 0}(\lambda r),& \text{ for }\kappa =0,\ \Im\lambda< 0
                                         \end{cases}
\end{align}
is the unique (up to a constant factor) non-trivial solution of \eqref{spectral equation}
which is square integrable at infinity.
\end{lem}

\begin{proof}
According to 13.7.2 and 13.2.4 in \cite{dlmf}, the following asymptotics hold true for large $r$ provided $\arg\lambda \in (-3\pi/2, \pi/2)$:
\begin{align*}
\begin{split}
 &M(a, b, 2\mathrm i\lambda r) =\big(1 +O(r^{-1})\big)\begin{cases}
                                               \dfrac{\Gamma(b)\mathrm e^{\pm\mathrm i\pi a}(2\mathrm i\lambda r)^{-a}}{\Gamma(b -a)}, &\text{for }\pm\Re\lambda >0,\, \Im\lambda >0;\\
						\dfrac{\Gamma(b)\mathrm e^{2\mathrm i\lambda r}(2\mathrm i\lambda r)^{a -b}}{\Gamma(a)}, &\text{for } \Im\lambda <0,\\
                                              \end{cases}
\end{split}
\end{align*}
unless $a \in -\mathbb N_0$ or $a -b \in \mathbb N_0$ (which cases are not relevant for our calculation).
By 13.7.3 in \cite{dlmf}, $U(a, b, 2\mathrm i\lambda r) =(2\mathrm i\lambda r)^{-a}\big(1 +O(r^{-1})\big)$ for all $a$, $b\in \mathbb C$, $\lambda \in\mathbb C\setminus \overline{\mathrm i\mathbb R_+}$ and $r\gg 1$. This allows us to compute the asymptotics of $\Phi^{\nu, \kappa}_M(\lambda\,\cdot)$ and $\Phi^{\nu, \kappa}_U(\lambda\,\cdot)$ with $\lambda\in \mathbb C\setminus(\mathbb R \cup \mathrm i\mathbb R_+)$ for large positive values of the argument. Since \eqref{Dirac symbol} is in the limit point case at infinity (see the proof of Theorem \ref{t: self-adjoint realisations}), by Theorem 1.4 in \cite{Weidmann1971} for $\lambda\in \mathbb C\setminus(\mathbb R \cup \mathrm i\mathbb R_+)$ there exists a unique (up to a factor) solution $\Phi_\infty^{\nu, \kappa}(\lambda; \cdot)$ to \eqref{spectral equation} which is square integrable at infinity. To construct such a solution it is enough to find a non-trivial linear combination of $\Phi^{\nu, \kappa}_M(\lambda\,\cdot)$ 
and $\Phi^{\nu, \kappa}_U(\lambda\,\cdot)$ for which the coefficient at the non square integrable asymptotic term vanishes. The coefficients in the statement of the lemma are chosen to satisfy this condition.
\end{proof}

We fix a solution to the spectral equation \eqref{spectral equation} satisfying the boundary condition \eqref{real bc} or \eqref{imaginary bc}:
\begin{lem}\label{l:Phi_0}
 For $(\nu, \kappa, \theta)\in \mathfrak M$ every solution of \eqref{spectral equation} satisfying the boundary condition \eqref{real bc} (for $\beta \in\overline{\mathbb R_+}$ and $\kappa \neq0$) or \eqref{imaginary bc} (for $\beta \in\mathrm i\mathbb R_+$ or $\kappa =0$) with $\lambda\in\mathbb C\setminus \mathrm i\overline{\mathbb R_+}$ is proportional to
\begin{align}\label{Phi_0}
 \Phi_{0, \theta}^{\nu, \kappa}(\lambda; r) :=a_\theta^{\nu, \kappa}(\lambda)\Phi^{\nu, \kappa}_U(\lambda r) +b_\theta^{\nu, \kappa}(\lambda)\Phi^{\nu, \kappa}_M(\lambda r),
\end{align}
with
\begin{align}\label{a}
 a_\theta^{\nu, \kappa}(\lambda):= \begin{cases}
                           \lambda^\beta\cos\theta, &\text{ for }\beta \in\overline{\mathbb R_+}\text{ and }\kappa \neq0;\\
			   \lambda^\beta\mathrm e^{\mathrm i\theta}, &\text{ for }\beta \in\mathrm i\mathbb R_+\text{ and }\kappa \neq0;\\
			   \lambda^{-\mathrm i\nu}\mathrm e^{\mathrm i\theta}, &\text{ for }\kappa =0
                          \end{cases}
\end{align}
and
\begin{align}\label{b}
 b_\theta^{\nu, \kappa}(\lambda):= \begin{cases}
                           \lambda^{-\beta}\sin\theta, &\text{ for }\beta \in\mathbb R_+;\\
			   \sin\theta -\cos\theta\ln\lambda, &\text{ for }\beta =0\text{ and }\kappa \neq0;\\
			   \lambda^{-\beta}\mathrm e^{-\mathrm i\theta}, &\text{ for }\beta \in\mathrm i\mathbb R_+\text{ and }\kappa \neq0;\\
			   \lambda^{\mathrm i\nu}\mathrm e^{-\mathrm i\theta}, &\text{ for }\kappa =0.
                          \end{cases}
\end{align}
Here the analytic branches of powers and logarithms of $\lambda$ are fixed by the convention $\arg\lambda \in(-3\pi/2, \pi/2)$. For $\lambda\in\mathbb R\setminus\{0\}$ both components of $\Phi_{0, \theta}^{\nu, \kappa}$ are real-valued.
\end{lem}

\begin{proof}
By the last statement of Lemma \ref{Lemma solutions dirac equation} any solution to \eqref{spectral equation} is a linear combination of $\Phi^{\nu, \kappa}_U(\lambda\,\cdot)$ and $\Phi^{\nu, \kappa}_M(\lambda\,\cdot)$.
For $\beta \geqslant 1/2$ (hence $\theta =\pi/2$), the only admissible solutions are multiples of $\Phi^{\nu, \kappa}_M(\lambda\cdot)$. 
Otherwise, substituting the general solution into \eqref{real bc} (or \eqref{imaginary bc}, respectively), we conclude the statement of the lemma from the asymptotics \eqref{Phi_M at zero}, \eqref{Phi_U at zero} and \eqref{Psi_M}, \eqref{Psi_U}.
\end{proof}

\begin{lem}\label{l:spectrum of D}
 For any $(\nu, \kappa, \theta)\in \mathfrak M$ the operator $D^{\nu, \kappa}_\theta$ has no eigenvalues.
\end{lem}

\begin{proof}
Assume that $\lambda\in \mathbb R\setminus\{0\}$ is an eigenvalue of $D^{\nu, \kappa}_\theta$. Then the corresponding eigenfunction must be a linear combination of $\Phi^{\nu, \kappa}_M(\lambda\,\cdot)$ and $\Phi^{\nu, \kappa}_U(\lambda\,\cdot)$ by Lemma \ref{Lemma solutions dirac equation}. No such non-trivial linear combination, however, can be square integrable, as follows from  analysis of asymptotics at infinity similar to the one from the proof of Lemma \ref{l:c^kappa,nu}. Analogously, for $\lambda =0$ any solution of \eqref{spectral equation} is a linear combination of \eqref{Psi_M} and \eqref{Psi_U} which are not square integrable at infinity.
\end{proof}

\begin{lem}\label{l: Wronskian}
Let $(\nu, \kappa, \theta) \in\mathfrak M$. For any two solutions $f =\displaystyle \binom{f_1}{f_2}$, $g =\displaystyle \binom{g_1}{g_2}$ of \eqref{spectral equation} the function
\begin{align}\label{W}
 W\big[f, g\big](r):= f_1(r)g_2(r) -f_2(r)g_1(r)
\end{align}
does not depend on $r\in \mathbb R_+$. Moreover, the formulae
\begin{align}
 \begin{split}\label{W of good solutions}
 &W\big[\Phi_\infty^{\nu, \kappa}(\lambda; \cdot), \Phi_{0, \theta}^{\nu, \kappa}(\lambda; \cdot)\big]=\\ &W[\Phi_M^{\nu, \kappa}, \Phi_U^{\nu, \kappa}]\begin{cases}
  c^{\nu, \kappa}(\lambda)a_\theta^{\nu, \kappa}(\lambda) -b_\theta^{\nu, \kappa}(\lambda), &\text{ for }\kappa \neq0,\ \lambda\in \mathbb C\setminus (\mathbb R\cup\mathrm i\mathbb R_+);\\
  -b_\theta^{\nu, 0}(\lambda),& \text{ for }\kappa =0,\ \Im\lambda> 0,\ \Re \lambda\neq 0;\\
   a_\theta^{\nu, 0}(\lambda),& \text{ for }\kappa =0,\ \Im\lambda< 0
\end{cases}
 \end{split}
\end{align}
and
\begin{align}\label{W_MU}
 W[\Phi_M^{\nu, \kappa}, \Phi_U^{\nu, \kappa}] =\begin{cases}
  -2\beta\kappa(\kappa +\beta), &\text{ for }\beta\in \mathbb R_+\setminus\{-\kappa\};\\
  -\beta, &\text{ for }\beta =-\kappa\in \mathbb R_+;\\
  \nu,& \text{ for }\beta =0 \text{ and } \kappa \neq0;\\
  -2\kappa^2\nu\beta,& \text{ for }\beta\in \mathrm i\mathbb R_+\text{ and } \kappa \neq0;\\
  2\mathrm i,& \text{ for }\kappa =0
\end{cases}
\end{align}
hold with the coefficients defined in Lemmata \ref{l:c^kappa,nu} and \ref{l:Phi_0}.
\end{lem}

\begin{proof}
 The independence of \eqref{W} from $r\in\mathbb R_+$ follows from \eqref{spectral equation} by a straightforward computation. Hence the analyticity of $\Phi_M^{\nu, \kappa}$ and $\Phi_U^{\nu, \kappa}$ in $\mathbb C\setminus\mathrm i\overline{\mathbb R_+}$ implies that the function $W[\Phi_M^{\nu, \kappa}, \Phi_U^{\nu, \kappa}]$ is constant on $\mathbb C\setminus\mathrm i\overline{\mathbb R_+}$. Relation \eqref{W of good solutions} follows by substituting \eqref{Phi_infty} and \eqref{Phi_0} into \eqref{W}. At last, \eqref{W_MU} follows from \eqref{W} by passing to the limit $r\to +0$ and using Lemma \ref{lemma asymptotics at zero}.
\end{proof}

\begin{lem}\label{l: Green function}
Let $(\nu, \kappa, \theta) \in\mathfrak M$. For $\lambda \in \mathbb C\setminus(\mathbb R\cup \mathrm i\mathbb R_+)$, the Green function $G^{\nu, \kappa}_\theta: \mathbb R\times(\mathbb R_+)^2 \to \mathbb C^{2\times2}$ of $D^{\nu, \kappa}_\theta$, i.e. the integral kernel of $(D^{\nu, \kappa}_\theta -\lambda\mathbb I)^{-1}$,
is given by
\begin{align}\label{Green function}
 G^{\nu, \kappa}_\theta(\lambda; x, y) :=\frac1{W\big[\Phi_\infty^{\nu, \kappa}(\lambda; \cdot), \Phi_{0, \theta}^{\nu, \kappa}(\lambda; \cdot)\big]}
 \begin{cases}
  \Phi_{0, \theta}^{\nu, \kappa}(\lambda; x)\big(\Phi_\infty^{\nu, \kappa}(\lambda; y)\big)^\intercal, &\text{for }x <y;\\
  \Phi_\infty^{\nu, \kappa}(\lambda; x)\big(\Phi_{0, \theta}^{\nu, \kappa}(\lambda; y)\big)^\intercal, &\text{for }x \geqslant y.
 \end{cases}
\end{align}
\end{lem}

\begin{proof}
Fix $\lambda \in \mathbb C\setminus(\mathbb R\cup \mathrm i\mathbb R_+)$. For $f\in\mathsf L^2(\mathbb R_+, \mathbb C^2)$ and $x \in\mathbb R_+$ let
\begin{align}\begin{split}\label{resolvent in action}
 \big(R^{\nu, \kappa}_\theta(\lambda) f\big)(x) &:=\frac{\Phi_\infty^{\nu, \kappa}(\lambda; x)}{W\big[\Phi_\infty^{\nu, \kappa}(\lambda; \cdot), \Phi_{0, \theta}^{\nu, \kappa}(\lambda; \cdot)\big]}\int_0^x\big(\Phi_{0, \theta}^{\nu, \kappa}(\lambda; y)\big)^\intercal f(y)\mathrm dy \\&+\frac{\Phi_{0, \theta}^{\nu, \kappa}(\lambda; x)}{W\big[\Phi_\infty^{\nu, \kappa}(\lambda; \cdot), \Phi_{0, \theta}^{\nu, \kappa}(\lambda; \cdot)\big]}\int_x^\infty \big(\Phi_\infty^{\nu, \kappa}(\lambda; y)\big)^\intercal f(y)\mathrm dy.
\end{split}
\end{align}
Since $\Phi_{0, \theta}^{\nu, \kappa}(\lambda;\, \cdot)\in \mathsf L^2\big((0, x)\big)$ and $\Phi_\infty^{\nu, \kappa}(\lambda;\, \cdot)\in \mathsf L^2\big((x, \infty)\big)$, the integrals on the right hand side of \eqref{resolvent in action} converge and the map $f \mapsto\big(R^{\nu, \kappa}_\theta(\lambda) f\big)(x)$ is continuous on $\mathsf L^2(\mathbb R_+, \mathbb C^2)$.
It is easy to observe that for all $f\in \mathsf C^\infty_0(\mathbb R_+, \mathbb C^2)$ the function $R^{\nu, \kappa}_\theta(\lambda) f$
belongs to the domain of $D^{\nu, \kappa}_\theta$ and satisfies
\begin{align*}%\label{resolvent works}
(D^{\nu, \kappa}_\theta -\lambda\mathbb I)R^{\nu, \kappa}_\theta(\lambda) f =f.
\end{align*}
Now for arbitrary $f\in\mathsf L^2(\mathbb R_+, \mathbb C^2)$ let $(f_n)_{n\in\mathbb N}$ be a sequence of $\mathsf C^\infty_0(\mathbb R_+, \mathbb C^2)$-functions converging to $f$ in $\mathsf L^2(\mathbb R_+, \mathbb C^2)$.
Applying \eqref{resolvent in action} to $f_n$ and passing to the limit $n \to\infty$ we observe that $R^{\nu, \kappa}_\theta(\lambda) f_n\underset{n\to\infty}\longrightarrow R^{\nu, \kappa}_\theta(\lambda) f$ pointwise in $\mathbb R_+$. On the other hand, since $\lambda$ belongs to the resolvent set of $D^{\nu, \kappa}_\theta$,
\begin{align*}
R^{\nu, \kappa}_\theta(\lambda) f_n =(D^{\nu, \kappa}_\theta -\lambda\mathbb I)^{-1}f_n \overset{\mathsf L^2(\mathbb R_+, \mathbb C^2)}{\underset{n\to\infty}\longrightarrow} (D^{\nu, \kappa}_\theta -\lambda\mathbb I)^{-1}f.
\end{align*}
This implies that $(D^{\nu, \kappa}_\theta -\lambda\mathbb I)^{-1}f =R^{\nu, \kappa}_\theta(\lambda) f =\int_{\mathbb R_+}G^{\nu, \kappa}_\theta(\lambda; \cdot, y)f(y)\mathrm dy$.
\end{proof}

\begin{lem}\label{l:no poles}
For all $(\nu, \kappa, \theta)\in \mathfrak M$ with $\kappa \neq0$ the following functions are continuous and have no zeros (see Lemmata \ref{l:c^kappa,nu}, \ref{l:Phi_0} for definitions):
\begin{enumerate}
 \item $c^{\nu, \kappa}_{+,\pm}a_\theta^{\nu, \kappa}(\cdot) -b_\theta^{\nu, \kappa}(\cdot)$, on $(\pm\mathbb R_+) +\mathrm i\overline{\mathbb R_+}$;
 \item $c^{\nu, \kappa}_{-}a_\theta^{\nu, \kappa}(\cdot) -b_\theta^{\nu, \kappa}(\cdot)$, on $\big(\mathbb R -\mathrm i\overline{\mathbb R_+}\big)\setminus\{0\}$.
\end{enumerate}
\end{lem}

\begin{proof}
By Definitions \eqref{a}, \eqref{b} the functions in question are analytic in $\mathbb C\setminus \mathrm i\overline{\mathbb R_+}$. % and can be analytically continued through the cut $\mathrm i\mathbb R_+$. 
Moreover, according to \eqref{c^kappa,nu}, \eqref{W of good solutions} and \eqref{W_MU} they are proportional to $W\big[\Phi_\infty^{\nu, \kappa}(\lambda; \cdot), \Phi_{0, \theta}^{\nu, \kappa}(\lambda; \cdot)\big]$ in the interiors of $(\pm\mathbb R_+) +\mathrm i\overline{\mathbb R_+}$ and $\big(\mathbb R -\mathrm i\overline{\mathbb R_+}\big)\setminus\{0\}$, respectively, with non-vanishing coefficients. But then they cannot vanish there, since otherwise $\Phi_{0, \theta}^{\nu, \kappa}(\lambda; \cdot)$ would be proportional to $\Phi_\infty^{\nu, \kappa}(\lambda; \cdot)$ and thus an eigenfunction of the self-adjoint operator $D^{\nu, \kappa}_\theta$ with the eigenvalue $\lambda\not\in\mathbb R$.

It remains to prove that there are no zeros for $\lambda\in \mathbb R\setminus\{0\}$. Recall that by our convention from Lemma \ref{l:Phi_0}, for any $\gamma\in \mathbb C$
\begin{align*}
 \lambda^\gamma =\begin{cases}
                 |\lambda|^\gamma, &\text{for }\lambda >0;\\ \mathrm e^{-\mathrm i\pi\gamma}|\lambda|^\gamma, &\text{for }\lambda <0
                \end{cases}
\end{align*}
holds.

For $\beta \geqslant \pi/2$ we have $\theta =\pi/2$ and the statement follows immediately.

Suppose now that $\beta\in (0, 1/2)$. In this case it is enough to show that
\begin{align}\label{Imaginary conditions}
\Im (\mathrm e^{(\pm1 -1)\mathrm i\pi\beta}c^{\nu, \kappa}_{+,\pm}) \neq0\quad\text{and}\quad\Im (\mathrm e^{(\pm1 -1)\mathrm i\pi\beta}c^{\nu, \kappa}_{-}) \neq0. 
\end{align}
This relations are implied by \eqref{c_+,pm}, \eqref{c_-} and
\begin{align*}
 \Im\Big(\frac{\mathrm i2^{2\beta}\Gamma(\beta +\mathrm i\nu)\mathrm e^{\pm\mathrm i\pi\beta}\Gamma(-2\beta)}{\Gamma(2\beta)\Gamma(1 -\beta +\mathrm i\nu)}\Big)= -\dfrac{2^{2\beta -1}\big|\Gamma(\beta +\mathrm i\nu)\big|^2\mathrm e^{\pm\pi\nu}}{2\beta\Gamma^2(2\beta)},
\end{align*}
which follows from the properties of the gamma function (see 5.5.3 in \cite{dlmf}).

For $\beta =0$ it is enough to establish that
\begin{align}\label{Imaginary conditions 2}
\Im c^{\nu, \kappa}_{+,\pm} \neq \pi(1 \mp1)/2 \quad\text{and}\quad\Im c^{\nu, \kappa}_{-} \neq \pi(1 \mp1)/2,
\end{align}
which follows from \eqref{c_+,pm}, \eqref{c_-} and the relations
\begin{align*}
 \Gamma(1- \mathrm i\nu)\Gamma(\mathrm i\nu)\mathrm e^{\pm\pi\nu} &=-\mathrm i\pi\coth(\pi\nu) \mp\mathrm i\pi,\\
 \Im\Big(\psi(1 +\mathrm i\nu) +\frac{\mathrm i}{2\nu}\Big) &=\frac\pi2\coth(\pi\nu),
\end{align*}
see 5.5.3 and 5.4.18 in \cite{dlmf}.

For $\beta \in\mathrm i\mathbb R_+$, the claim follows from
\begin{align*}%\label{Imaginary conditions 3}
|c^{\nu, \kappa}_{+,\pm}| \neq \mathrm e^{(1 \mp1)\pi\mathrm i\beta} \quad\text{and}\quad |c^{\nu, \kappa}_{-}| \neq \mathrm e^{(1 \mp1)\pi\mathrm i\beta},
\end{align*}
implied by \eqref{c_+,pm}, \eqref{c_-} and
\begin{align*}
 |c^{\nu, \kappa}_{+,\pm}|^2 =\mathrm e^{(2 \mp4)\pi\mathrm i\beta}\frac{\sinh\big(\pi(\nu -\mathrm i\beta)\big)}{\sinh\big(\pi(\nu +\mathrm i\beta)\big)}, \quad |c^{\nu, \kappa}_{-}|^2 =\mathrm e^{2\pi\mathrm i\beta}\frac{\sinh\big(\pi(\nu +\mathrm i\beta)\big)}{\sinh\big(\pi(\nu -\mathrm i\beta)\big)},
\end{align*}
see 5.4.3 in \cite{dlmf}.
\end{proof}

\begin{lem}\label{l: boundedness of G on compacta}
Let $I$, $J$ be compact subsets of $\mathbb R_+$ and $\mathbb R\setminus\{0\}$, respectively. Then for every $x >0$ there exist $C^{\nu, \kappa}_{\theta, \pm}(I, J; x)\in \mathbb R_+$ such that
\begin{align*}
 \sup_{(\lambda, y) \in (J\pm \mathrm i(0, 1])\times I}\big\|G^{\nu, \kappa}_\theta(\lambda; x, y)\big\|_{\mathbb C^{2\times 2}} \leqslant C^{\nu, \kappa}_{\theta, \pm}(I, J; x)
\end{align*}
holds.
\end{lem}

\begin{proof}
The functions $\Phi^{\nu, \kappa}_M$ and $\Phi^{\nu, \kappa}_U$ are analytic in $\mathbb C\setminus \mathrm i\overline{\mathbb R_+}$. Hence by \eqref{Green function}, \eqref{W of good solutions}, \eqref{W_MU}, and Lemmata \ref{l:c^kappa,nu}, \ref{l:Phi_0} and \ref{l:no poles} for every $x >0$ the function $G^{\nu, \kappa}_\theta(\cdot; x, \cdot)$ allows a unique continuous extension from $\big(J\pm \mathrm i(0, 1]\big)\times I$ to the compact $\big(J\pm \mathrm i[0, 1]\big)\times I$. The statement of the lemma follows.
\end{proof}

\begin{lem}\label{t: spectral function}
Consider the map $E^{\nu, \kappa}_\theta: \mathbb R\times \mathbb R_+^2\to \mathbb C^{2\times 2}$,
\begin{align*}
 E^{\nu, \kappa}_\theta(\lambda; x,y) :=\frac1{2\pi\mathrm i}\big(G^{\nu, \kappa}_\theta(\lambda +\mathrm i0; x,y) -G^{\nu, \kappa}_\theta(\lambda -\mathrm i0; x,y)\big).%\label{E}
\end{align*}
Then
\begin{align}\label{E simplified}
 E^{\nu, \kappa}_\theta(\lambda; x,y) =m^{\nu, \kappa}_\theta(\lambda)\Phi_{0, \theta}^{\nu, \kappa}(\lambda; x)\big(\Phi_{0, \theta}^{\nu, \kappa}(\lambda; y)\big)^\intercal
\end{align}
holds with
\begin{align}\begin{split}\label{m}
 &m^{\nu, \kappa}_\theta(\lambda):= \\ &\dfrac{\big(c^{\nu, \kappa}(\lambda -\mathrm i0) -c^{\nu, \kappa}(\lambda +\mathrm i0)\big)}{2\pi\mathrm iW[\Phi_M^{\nu, \kappa}, \Phi_U^{\nu, \kappa}]\big(c^{\nu, \kappa}(\lambda +\mathrm i0)a_\theta^{\nu, \kappa}(\lambda) -b_\theta^{\nu, \kappa}(\lambda)\big)\big(c^{\nu, \kappa}(\lambda -\mathrm i0)a_\theta^{\nu, \kappa}(\lambda) -b_\theta^{\nu, \kappa}(\lambda)\big)},
\end{split}
\end{align}
for $\kappa \neq0$ and 
\begin{align}\label{m^0}
m^{\nu, 0}_\theta(\lambda) :=(4\pi)^{-1}.
\end{align}
For any bounded interval $\mathcal I\subset \mathbb R\setminus\{0\}$ the spectral projection on $\mathcal I$ of $D^{\nu, \kappa}_\theta$ is given by
\begin{align}\label{spectral projection on intervals}
P_{\mathcal I}(D^{\nu, \kappa}_\theta)f =\int_{\mathcal I}\int_0^\infty E^{\nu, \kappa}_\theta(\lambda; \cdot,y)f(y)\,\mathrm dy\,\mathrm d\lambda
\end{align}
for every compactly supported $f\in \mathsf L^2(\mathbb R_+, \mathbb C^2)$.
\end{lem}

\begin{proof}
The representations \eqref{E simplified} and \eqref{m} follow from Lemma \ref{l: Green function} together with \eqref{W of good solutions}, \eqref{Phi_0} and \eqref{Phi_infty}.
By Lemma \ref{l:spectrum of D} $D^{\nu, \kappa}_\theta$ has no eigenvalues. Hence for any compactly supported $f\in \mathsf L^2(\mathbb R_+, \mathbb C^2)$ by the Stone's formula (see e.g. Theorem 4.3 in \cite{TeschlMathematicalMethods}) and Lemma \ref{l: Green function} we have
\begin{align*}\begin{split}
P_{\mathcal I}(D^{\nu, \kappa}_\theta)f =\lim_{\varepsilon\to +0}\int_{\mathcal I}\int_0^\infty\frac{G^{\nu, \kappa}_\theta(\lambda +\mathrm i\varepsilon; \cdot, y) -G^{\nu, \kappa}_\theta(\lambda -\mathrm i\varepsilon; \cdot, y)}{2\pi\mathrm i}f(y)\,\mathrm dy\,\mathrm d\lambda.
\end{split}\end{align*}
By Lemma \ref{l: boundedness of G on compacta} and dominated convergence we can interchange the limit and the integration obtaining \eqref{spectral projection on intervals}.
\end{proof}

\begin{proof}[Proof of Theorem \ref{t:spectral representation}]
For $n\in\mathbb N$ let $\mathcal E_n := (-n, -n^{-1})\cup (n^{-1}, n)$. For compactly supported $f\in \mathsf L^2(\mathbb R_+, \mathbb C^2)$ we observe that by Lemma \ref{l:spectrum of D}, Fubini's theorem, \eqref{spectral projection on intervals} and \eqref{E simplified}
\begin{align*}\begin{split}
 \|f\|_{\mathsf L^2(\mathbb R_+, \mathbb C^2)}^2 &=\lim_{n\to\infty}\big\|P_{\mathcal E_n}(D^{\nu, \kappa}_\theta)f\big\|_{\mathsf L^2(\mathbb R_+, \mathbb C^2)}^2 =\lim_{n\to\infty}\langle f, P_{\mathcal E_n}(D^{\nu, \kappa}_\theta)f\rangle_{\mathsf L^2(\mathbb R_+, \mathbb C^2)}\\ &=\bigg\|\sqrt{m^{\nu, \kappa}_\theta(\cdot)}\int_0^\infty \big(\Phi_{0, \theta}^{\nu, \kappa}(\cdot; y)\big)^\intercal f(y)\,\mathrm dy\bigg\|_{\mathsf L^2(\mathbb R, \mathbb C)}^2
\end{split}\end{align*}
holds. Thus $\mathcal U^{\nu, \kappa}_\theta$ is well defined by \eqref{U} and is isometric, i.e. $(\mathcal U^{\nu, \kappa}_\theta)^*\mathcal U^{\nu, \kappa}_\theta =\mathbb I_{\mathsf L^2(\mathbb R_+, \mathbb C^2)}$.

Now for any $f\in\mathfrak D(D^{\nu, \kappa}_\theta)$ integrating by parts with the help of Theorem \ref{t: self-adjoint realisations} we obtain
\begin{align}\begin{split}\label{U D intertwining}
     &\mathcal U^{\nu, \kappa}_\theta D^{\nu, \kappa}_\theta f =\underset{R\to\infty}\Ltwolim\bigg(\sqrt{m^{\nu, \kappa}_\theta(\cdot)}\int_{1/R}^R \big(\Phi_{0, \theta}^{\nu, \kappa}(\cdot; y)\big)^\intercal D^{\nu, \kappa}_\theta f(y)\,\mathrm dy\bigg) \\&= \underset{R\to\infty}\Ltwolim \bigg((\cdot)\sqrt{m^{\nu, \kappa}_\theta(\cdot)}\int_{1/R}^R \big(\Phi_{0, \theta}^{\nu, \kappa}(\cdot; y)\big)^\intercal f(y)\,\mathrm dy\bigg) =\Lambda\,\mathcal U^{\nu, \kappa}_\theta f,
             \end{split}
\end{align}
hence \eqref{spectral representation}.

It remains to prove that $\mathcal U^{\nu, \kappa}_\theta$ is surjective. We denote by $P_{\theta}^{\nu, \kappa}$ the orthogonal projector in $\mathsf L^2(\mathbb R, \mathbb C)$ onto 
the range of $\mathcal U^{\nu, \kappa}_\theta$. According to \eqref{U D intertwining} the operator 
\begin{align*}
                 \Lambda^{\nu, \kappa}_\theta : \mathcal U^{\nu, \kappa}_\theta\mathfrak{D}(D^{\nu, \kappa}_\theta)
                 \rightarrow P_{\theta}^{\nu, \kappa}\mathsf{L}^2(\mathbb{R},\mathbb{C});\quad \Lambda^{\nu, \kappa}_\theta g:=
                 \Lambda g
\end{align*} 
is well-defined and self-adjoint in $P_{\theta}^{\nu, \kappa}\mathsf{L}^2(\mathbb{R},\mathbb{C})$. Since for every $g\in \mathfrak{D}(\Lambda)$ and $h\in \mathfrak{D}(\Lambda_{\theta}^{\nu, \kappa})$
\begin{align*}
\langle \Lambda_{\theta}^{\nu, \kappa}h, P_{\theta}^{\nu, \kappa}g \rangle_{\mathsf L^2(\mathbb R, \mathbb C)} = \langle h, \Lambda g \rangle_{\mathsf L^2(\mathbb R, \mathbb C)}
\end{align*}
holds, we conclude that $P_{\theta}^{\nu, \kappa}\mathfrak{D}(\Lambda)\subset\mathfrak{D}\big((\Lambda_{\theta}^{\nu, \kappa})^*\big) =\mathfrak{D}(\Lambda_{\theta}^{\nu, \kappa})$ and $\Lambda P_{\theta}^{\nu, \kappa}\mathfrak{D}(\Lambda) \subset P_{\theta}^{\nu, \kappa}\mathsf{L}^2(\mathbb{R},\mathbb{C})$, i.e., the range of $\mathcal U^{\nu, \kappa}_\theta$ is a reducing subspace of $\Lambda$.
Hence there exists a measurable $\Omega\subset \mathbb R$ such that 
\begin{align}\label{Omega exclusion}
\mathcal U^{\nu, \kappa}_\theta\mathsf L^2(\mathbb R_+, \mathbb C^2) =
\Ran\big(\mathds{1}_{\mathbb R\setminus\Omega}(\Lambda)\big)
\end{align}
(combine Corollary 4.6 and Theorem 4.8 in \cite{TeschlMathematicalMethods}).

Our goal is to show that $\Omega$ is a Lebesgue null set. If this is not the case, then there exists $\delta >0$ so that $\Omega\cap\big([-\delta^{-1}, -\delta]\cup[\delta, \delta^{-1}]\big)$ is of positive Lebesgue measure. Since $\Phi_{0, \theta}^{\nu, \kappa}(\cdot; 1)$ is analytic in $\mathbb C\setminus \mathrm i\overline{\mathbb R_+}$ and real-valued on $\mathbb R\setminus \{0\}$, the set
\[\Xi_0 :=\big\{\lambda\in \mathbb R\setminus \{0\}: (\Phi_{0, \theta}^{\nu, \kappa})_1(\lambda; 1) =0\big\}\]
where the first component of $\Phi_{0, \theta}^{\nu, \kappa}$ vanishes is Lebesgue null, and thus at least one of the sets
\begin{align}\label{positive sets}
 \Xi_\pm :=\big\{\lambda\in \mathbb R\setminus \{0\}: \pm(\Phi_{0, \theta}^{\nu, \kappa})_1(\lambda; 1) >0\big\}\cap \Omega\cap\big([-\delta^{-1}, -\delta]\cup[\delta, \delta^{-1}]\big)
\end{align}
has positive Lebesgue measure. Without restriction suppose that $\Xi_+$ is not null.
A simple calculation gives
\begin{align}\label{U^*}
 (\mathcal U^{\nu, \kappa}_\theta)^* g =\underset{R\to\infty}\Ltwolim\int_{-R}^R \Phi_{0, \theta}^{\nu, \kappa}(\lambda; \cdot) g(\lambda) \sqrt{m^{\nu, \kappa}_\theta(\lambda)}\,\mathrm d\lambda
\end{align}
for any $g\in \mathsf L^2(\mathbb R, \mathbb C)$. Choosing $g :=\mathds1_{\Xi_+}$% (the indicator function of $\Xi_+$) 
and using $m^{\nu, \kappa}_\theta(\cdot)>0$ we obtain
\begin{align}\label{positive at 1}
\big((\mathcal U^{\nu, \kappa}_\theta)^* \mathds1_{\Xi_+}\big)_1(1) >0.
\end{align}
By dominated convergence \eqref{U^*} implies that the map $r \mapsto\big((\mathcal U^{\nu, \kappa}_\theta)^* \mathds1_{\Xi_+}\big)_1(r)$ is continuous at $r =1$, thus by \eqref{positive at 1} we get $(\mathcal U^{\nu, \kappa}_\theta)^* \mathds1_{\Xi_+} \neq 0$ in $\mathsf L^2(\mathbb R_+, \mathbb C^2)$. This is, however, not possible, since by \eqref{positive sets} $\mathds1_{\Xi_+}\in \big(\mathcal U^{\nu, \kappa}_\theta\mathsf L^2(\mathbb R_+, \mathbb C^2)\big)^\perp =\ker\big((\mathcal U^{\nu, \kappa}_\theta)^*\big)$.
The contradiction implies that $\Omega$ is null, thus the right hand side of \eqref{Omega exclusion} coincides with $\mathsf L^2(\mathbb R, \mathbb C)$.
\end{proof}

\section{Proof of Theorem \ref{t: the alternative}}

\begin{proof}[Proof of Theorem \ref{t: the alternative}, part \emph{I}]
 We use $\varphi_n := (\mathcal U^{\nu, \kappa}_\theta)^*\big((\cdot)^{-1}\mathbbm 1_{[1/n^2, 1/n]}(\cdot)\big)$ for $n\in\mathbb N$ big enough as a test function. Theorem \ref{t:spectral representation} implies that $\varphi_n\in P_{\theta, \infty}^{\nu, \kappa}\mathfrak D(D^{\nu, \kappa}_\theta) \subset P_{\theta, \infty}^{\nu, \kappa}\mathfrak D\big(|D^{\nu, \kappa}_\theta|^{1/2}\big)$ and
\begin{align}\begin{split}\label{kinetic energy}
 &\langle|D^{\nu, \kappa}_\theta|^{1/2}\varphi_n, |D^{\nu, \kappa}_\theta|^{1/2}\varphi_n\rangle_{\mathsf L^2(\mathbb R_+, \mathbb C^2)}\\ &=\langle \mathcal U^{\nu, \kappa}_\theta\varphi_n, \big(\mathcal U^{\nu, \kappa}_\theta D^{\nu, \kappa}_\theta(\mathcal U^{\nu, \kappa}_\theta)^*\big)\mathcal U^{\nu, \kappa}_\theta\varphi_n\rangle_{\mathsf L^2(\mathbb R, \mathbb C^2)} =\int_{1/n^2}^{1/n}\lambda^{-1}\,\mathrm d\lambda =\ln n.
\end{split}
\end{align}

We also have by \eqref{U^*}
\begin{align}\begin{split}\label{V form}
 \langle\varphi_n, V\varphi_n\rangle_{\mathsf L^2(\mathbb R_+, \mathbb C^2)}= \int_0^\infty\int_{1/n^2}^{1/n}\int_{1/n^2}^{1/n}&\frac{\sqrt{m^{\nu, \kappa}_\theta(\lambda)}}\lambda\frac{\sqrt{m^{\nu, \kappa}_\theta(\mu)}}\mu\\ &\times\big\langle\Phi_{0, \theta}^{\nu, \kappa}(\lambda; r), V(r)\Phi_{0, \theta}^{\nu, \kappa}(\mu; r)\big\rangle_{\mathbb C^2}\mathrm d\mu\,\mathrm d\lambda\, \mathrm dr.
\end{split}
\end{align}
Recalling Lemma \ref{l:Phi_0} together with \eqref{Phi_M at zero} and \eqref{Phi_U at zero} and taking into account the boundedness of $\Phi_U^{\nu, \kappa}$ and $\Phi_M^{\nu, \kappa}$ on $[1, \infty)$ for all $\lambda, r\in \mathbb R_+$ (see Lemma \ref{Lemma solutions dirac equation} together with 13.7.2 and 13.7.3 in \cite{dlmf})  we obtain the decomposition
\begin{align*}%\label{A and B}
 \Phi_{0, \theta}^{\nu, \kappa}(\lambda; r) =A_{\theta}^{\nu, \kappa}(r) +B_{\theta}^{\nu, \kappa}(\lambda; r),
\end{align*}
where $A_{\theta}^{\nu, \kappa}$ is defined in \eqref{A_theta}
and
\begin{align}\label{B estimate}
 \big\|B_{\theta}^{\nu, \kappa}(\lambda; r)\big\|_{\mathbb C^2} \leqslant C^{\nu, \kappa}\lambda r^{1- \Re\beta}
\end{align}
with some finite $C^{\nu, \kappa} >0$. Thus writing $V =|V|^{1/2}(\sign V)|V|^{1/2}$ and using the Cauchy inequality we can estimate
\begin{align*}
 &\Big|\big\langle\Phi_{0, \theta}^{\nu, \kappa}(\lambda; r), V(r)\Phi_{0, \theta}^{\nu, \kappa}(\mu; r)\big\rangle_{\mathbb C^2}\Big|\\ &\leqslant 2\big\langle A_{\theta}^{\nu, \kappa}(r), \big|V(r)\big|A_{\theta}^{\nu, \kappa}(r)\big\rangle_{\mathbb C^2} +(C^{\nu, \kappa})^2(\lambda^2 +\mu^2)\big\|V(r)\big\|_{\mathbb C^{2\times2}} r^{2- 2\Re\beta},
\end{align*}
where the right hand side is integrable in $r$ over $\mathbb R_+$ by \eqref{V integrability}. Hence the order of integrations in \eqref{V form} can be interchanged by Fubini's theorem.
By Schwarz inequality and \eqref{B estimate} we have
\begin{align}\begin{split}\label{A +B CS}
 &\Re\int_0^\infty\big\langle\Phi_{0, \theta}^{\nu, \kappa}(\lambda; r), V(r)\Phi_{0, \theta}^{\nu, \kappa}(\mu; r)\big\rangle_{\mathbb C^2} \mathrm dr\\ &=\Re\big\langle|V|^{1/2}\big(A_{\theta}^{\nu, \kappa} +B_{\theta}^{\nu, \kappa}(\lambda; \cdot)\big), (\sign V)|V|^{1/2}\big(A_{\theta}^{\nu, \kappa} +B_{\theta}^{\nu, \kappa}(\mu; \cdot)\big)\big\rangle_{\mathsf L^2(\mathbb R_+, \mathbb C^2)}\\ &\geqslant \frac12\int_0^\infty\big\langle A_\theta^{\nu, \kappa}(r), V(r)A_\theta^{\nu, \kappa}(r)\big\rangle_{\mathbb C^{2}}\mathrm dr\\ &-\bigg(\frac{\int_0^\infty\big\langle A_\theta^{\nu, \kappa}(r), \big|V(r)\big|A_\theta^{\nu, \kappa}(r)\big\rangle_{\mathbb C^{2}}\mathrm dr}{\int_0^\infty\big\langle A_\theta^{\nu, \kappa}(r), V(r)A_\theta^{\nu, \kappa}(r)\big\rangle_{\mathbb C^{2}}\mathrm dr} +\frac12\bigg)\\ &\hspace{4.2cm}\times(C^{\nu, \kappa})^2(\lambda^2 +\mu^2)\int_0^\infty\big\|V(r)\big\|_{\mathbb C^{2\times2}}r^{2 -2\Re\beta}\mathrm dr.
\end{split}\end{align}
Inserting \eqref{A +B CS} into \eqref{V form} we arrive at
\begin{align}\begin{split}\label{V form estimate}
&\langle\varphi_n, V\varphi_n\rangle_{\mathsf L^2(\mathbb R_+, \mathbb C^2)} \geqslant \frac12\int_0^\infty\big\langle A_\theta^{\nu, \kappa}(r), V(r)A_\theta^{\nu, \kappa}(r)\big\rangle_{\mathbb C^{2}}\mathrm dr\bigg(\int_{1/n^2}^{1/n}\frac{\sqrt{m^{\nu, \kappa}_\theta(\lambda)}}\lambda\mathrm d\lambda\bigg)^2\\ &-\bigg(\frac{2\int_0^\infty\big\langle A_\theta^{\nu, \kappa}(r), \big|V(r)\big|A_\theta^{\nu, \kappa}(r)\big\rangle_{\mathbb C^{2}}\mathrm dr}{\int_0^\infty\big\langle A_\theta^{\nu, \kappa}(r), V(r)A_\theta^{\nu, \kappa}(r)\big\rangle_{\mathbb C^{2}}\mathrm dr} +1\bigg)(C^{\nu, \kappa})^2\\ &\quad \times\int_0^\infty \big\|V(r)\big\|_{\mathbb C^{2\times2}}r^{2 -2\Re\beta}\mathrm dr\bigg(\int_{1/n^2}^{1/n}\lambda\sqrt{m^{\nu, \kappa}_\theta(\lambda)}\mathrm d\lambda\bigg)\bigg(\int_{1/n^2}^{1/n}\frac{\sqrt{m^{\nu, \kappa}_\theta(\lambda)}}\lambda\mathrm d\lambda\bigg),
\end{split}
\end{align}
which is positive for $n$ big enough due to \eqref{V integrability} and \eqref{VL via A}.
It follows from \eqref{m}, \eqref{m^0}, \eqref{W_MU} and Lemmata \ref{l:c^kappa,nu} and \ref{l:Phi_0}, that there exists $h^{\nu, \kappa}> 0$ such that for $(\nu, \kappa, \theta)\in \mathfrak M_{\mathrm I}$
\begin{align*}
 m^{\nu, \kappa}_\theta(\lambda)\geqslant h^{\nu, \kappa}\lambda^{-2\Re\beta}
\end{align*}
holds for all $\lambda >0$. This implies
\begin{align}\label{m coefficient}
 \int_{1/n^2}^{1/n}\frac{\sqrt{m^{\nu, \kappa}_\theta(\lambda)}}\lambda \mathrm d\lambda \geqslant h^{\nu, \kappa}\begin{cases}
                                                                                                                   \ln n, &\text{for }\beta \in\overline{\mathrm i\mathbb R_+};\\
														   \dfrac{n^{2\beta} -n^{\beta}}{\beta}, &\text{for }\beta \in(0, 1/2).
                                                                                                                  \end{cases}
\end{align}
Now \eqref{kinetic energy}, \eqref{V form estimate} and \eqref{m coefficient} imply that the quadratic form of $D_{\theta, \infty}^{\nu, \kappa}(V)$ computed on $\varphi_n$ becomes negative for $n$ big enough. The existence of negative spectrum of $D_{\theta, \infty}^{\nu, \kappa}(V)$ follows by the minimax principle.
\end{proof}

For $E \in (0, \infty]$ let $P_{\theta, E}^{\nu, \kappa}:= P_{[0, E)}(D^{\nu, \kappa}_\theta)$ be the spectral projector of $D^{\nu, \kappa}_\theta$ corresponding to the interval $[0, E)$. 
Under Hypothesis A the quadratic form
\begin{align}\label{d form}
\mathfrak d_{\theta, E}^{\nu, \kappa}(V)[\cdot] :=\big\||D^{\nu, \kappa}_\theta|^{1/2}\cdot\big\|_{\mathsf L^2(\mathbb R_+, \mathbb C^2)}^2 -\langle |V|^{1/2}\cdot, (\sign V)|V|^{1/2}\cdot\rangle_{\mathsf L^2(\mathbb R_+, \mathbb C^2)}
\end{align}
is closed and bounded from below on $P_{\theta, E}^{\nu, \kappa}\mathfrak D\big(|D^{\nu, \kappa}_\theta|^{1/2}\big)$. We define
\begin{align*}%\label{D projected with V}
 D_{\theta, E}^{\nu, \kappa}(V) :=P_{\theta, E}^{\nu, \kappa}(D^{\nu, \kappa}_\theta -V)P_{\theta, E}^{\nu, \kappa}
\end{align*}
as the self-adjoint operator in $P_{\theta, E}^{\nu, \kappa}\mathsf L^2(\mathbb R_+, \mathbb C^2)$ corresponding to $\mathfrak d_{\theta, E}^{\nu, \kappa}(V)$.

For $\tau >0$ Hypothesis A implies that the Birman-Schwinger operator
\begin{align}\label{BS operator}
 B_{\theta, E}^{\nu, \kappa}(V, \tau) :=(D^{\nu, \kappa}_\theta P_{\theta, E}^{\nu, \kappa} +\tau\mathbb I)^{-1/2}P_{\theta, E}^{\nu, \kappa}V(D^{\nu, \kappa}_\theta P_{\theta, E}^{\nu, \kappa} +\tau\mathbb I)^{-1/2}P_{\theta, E}^{\nu, \kappa}
\end{align}
is bounded in $\mathsf L^2(\mathbb R_+, \mathbb C^2)$.

We will use the following version of the Birman-Schwinger principle:

\begin{lem}\label{l:Birman-Schwinger}
The equality
\begin{align}\label{Birman-Schwinger}\begin{split}
 &\rank P_{(-\infty, -\tau)}\big(D_{\theta, E}^{\nu, \kappa}(V)\big) = \rank P_{(1, \infty)}\big(B_{\theta, E}^{\nu, \kappa}(V, \tau)\big)\end{split}
\end{align}
holds for any $\tau >0$ with $B_{\theta, E}^{\nu, \kappa}(V, \tau)$ defined in \eqref{BS operator}.
\end{lem}

\begin{proof}
For any $\varphi\in P_{\theta, E}^{\nu, \kappa}\mathfrak D\big(|D^{\nu, \kappa}_\theta|^{1/2}\big)$, $\tau >0$ we have
\begin{align*}\begin{split}%\label{lower form bound}
 &\mathfrak d_{\theta, E}^{\nu, \kappa}(V)[\varphi] +\tau\|\varphi\|_{\mathsf L^2(\mathbb R_+, \mathbb C^2)}^2\\ & =\big\langle(D^{\nu, \kappa}_\theta P_{\theta, E}^{\nu, \kappa} +\tau\mathbb I)^{1/2}\varphi, \big(\mathds1 -B_{\theta, E}^{\nu, \kappa}(\tau)\big)(D^{\nu, \kappa}_\theta P_{\theta, E}^{\nu, \kappa} +\tau\mathbb I)^{1/2}\varphi\big\rangle_{\mathsf L^2(\mathbb R_+, \mathbb C^2)}.
\end{split}\end{align*}
The identity \eqref{Birman-Schwinger} follows by the minimax principle.
\end{proof}

\begin{lem}\label{l:IntegralEstimate}
For $q \geqslant1$ and $\tau >0$ the estimate
\begin{align}\begin{split}\label{rank via integral}
 &\rank P_{(-\infty, -\tau)}\big(D_{\theta, E}^{\nu, \kappa}(V)\big)\\ &\leqslant \int_0^\infty\int_0^E m^{\nu, \kappa}_\theta(\lambda)(\lambda +\tau)^{-q}\big\|V_+^{q/2}(r)\Phi_{0, \theta}^{\nu, \kappa}(\lambda; r)\big\|_{\mathbb C^2}^2\mathrm d\lambda\,\mathrm dr\end{split}
\end{align}
holds. If the integral on the right hand side of \eqref{rank via integral} is finite with $E :=\infty$, then Hypothesis A is satisfied for $V :=V_+$.
\end{lem}

\begin{proof}
For  $q \geqslant1$ let $\mathfrak S^q$ be the $q$th Schatten-von-Neumann class of compact operators. Then for any self-adjoint, non-negative operators $A$ and $B$ with $A^qB^q\in \mathfrak S^{2}$ we have $AB\in \mathfrak S^{2q}$ and $\|AB\|_{\mathfrak S^{2q}}^{2q}\leqslant \|A^qB^q\|_{\mathfrak S^{2}}^2$ (see Appendix B in \cite{LiebThirring1976}). Hence by \eqref{U^*}
\begin{align}\label{integral estimate}
\begin{split}
 &\big\|V_+^{1/2}(D^{\nu, \kappa}_\theta P_{\theta, E}^{\nu, \kappa} +\tau\mathbb I)^{-1/2}P_{\theta, E}^{\nu, \kappa}\big\|_{\mathfrak S^{2q}}^{2q}\\ &\leqslant\big\|V_+^{q/2}(\mathcal U^{\nu, \kappa}_\theta)^*\mathcal U^{\nu, \kappa}_\theta(D^{\nu, \kappa}_\theta P_{\theta, E}^{\nu, \kappa} +\tau\mathbb I)^{-q/2}P_{\theta, E}^{\nu, \kappa}(\mathcal U^{\nu, \kappa}_\theta)^*\big\|_{\mathfrak S^{2}}^2\\ &=\big\|V_+^{q/2}(\mathcal U^{\nu, \kappa}_\theta)^*\mathbbm 1_{[0, E)}(\cdot)(\cdot +\tau)^{-q/2}\big\|_{\mathfrak S^{2}}^2\\ &=\int_0^\infty\int_0^E m^{\nu, \kappa}_\theta(\lambda)(\lambda +\tau)^{-q}\big\|V_+^{q/2}(r)\Phi_{0, \theta}^{\nu, \kappa}(\lambda; r)\big\|_{\mathbb C^2}^2\mathrm d\lambda\,\mathrm dr.
\end{split}
\end{align}
Estimating the right hand side of \eqref{Birman-Schwinger} from above by
\begin{align*}
\begin{split}
 &\big\|(D^{\nu, \kappa}_\theta P_{\theta, E}^{\nu, \kappa} +\tau\mathbb I)^{-1/2}P_{\theta, E}^{\nu, \kappa}V_+(D^{\nu, \kappa}_\theta P_{\theta, E}^{\nu, \kappa} +\tau\mathbb I)^{-1/2}P_{\theta, E}^{\nu, \kappa}\big\|_{\mathfrak S^{q}}^{q}\\ &=\big\|V_+^{1/2}(D^{\nu, \kappa}_\theta P_{\theta, E}^{\nu, \kappa} +\tau\mathbb I)^{-1/2}P_{\theta, E}^{\nu, \kappa}\big\|_{\mathfrak S^{2q}}^{2q}
\end{split}
\end{align*}
and applying \eqref{integral estimate} we conclude \eqref{rank via integral}.

If the integral on the right hand side of \eqref{rank via integral} is finite with $E :=\infty$, then \eqref{integral estimate} implies that $P_{\theta, \infty}^{\nu, \kappa}V_+P_{\theta, \infty}^{\nu, \kappa}$ is a form compact perturbation of $D^{\nu, \kappa}_\theta $ in $P_{\theta, \infty}^{\nu, \kappa}\mathsf L^2(\mathbb R_+, \mathbb C^2)$. Hypothesis A follows with standard arguments.
\end{proof}

\begin{proof}[Proof of Theorem \ref{t: the alternative}, part \emph{II}]
Applying Lemma \ref{l:IntegralEstimate} and passing to $\tau \to +0$ in \eqref{rank via integral} we obtain
\begin{align}\begin{split}\label{integral estimate applied}
 &\rank P_{(-\infty, 0)}\big(D_{\theta, \infty}^{\nu, \kappa}(V)\big)\\ &\leqslant \int_0^\infty \big\|V_+^{q}(r)\big\|_{\mathbb C^{2\times 2}}\int_0^\infty \lambda^{-q}m^{\nu, \kappa}_\theta(\lambda)\big\|\Phi_{0, \theta}^{\nu, \kappa}(\lambda; r)\big\|_{\mathbb C^2}^2\mathrm d\lambda\,\mathrm dr.\end{split}
\end{align}

In the case $\beta\in (0, 1/2)$ and $\theta\in (0, \pi)\setminus\{\pi/2\}$, we rescale the variable $\mu:= |\cot\theta|^{1/(2\beta)}\lambda$ and observe
\begin{align}\label{integral estimate first use}
 &\int_0^\infty \lambda^{-q}m_\theta^{\nu, \kappa}(\lambda)\big\|\Phi_{0, \theta}^{\nu, \kappa}(\lambda; r)\big\|_{\mathbb C^2}^2\mathrm d\lambda\\ &=\frac{|\cot\theta|^{(q -2\beta -1)/(2\beta)}}{\sin^2\theta}\int_0^\infty\mu^{-q +2\beta}\widetilde m_{\theta}^{\nu, \kappa}(\mu)\Big\|\Phi_{0, \theta}^{\nu, \kappa}\big(|\cot\theta|^{-1/(2\beta)}\mu; r\big)\Big\|_{\mathbb C^2}^2\mathrm d\mu,
\end{align}
where 
\begin{align*}
\widetilde m^{\nu, \kappa}_\theta(\mu) :=\sin^2\theta|\cot\theta|\mu^{-2\beta}m^{\nu, \kappa}_\theta\big(|\cot\theta|^{-1/(2\beta)}\mu\big)
\end{align*}
satisfies
\begin{align}\label{tilde m bound}
\big|\widetilde m_\theta^{\nu, \kappa}(\mu)\big| \leqslant C^{\nu, \kappa}\begin{cases}
                                  \mu^{-4\beta}, & \text{for }\mu \geqslant 1;\\ 1, & \text{for }\mu \leqslant 1
                                 \end{cases}
\end{align}
with some $C^{\nu, \kappa} >0$ independent from $\theta$ (see \eqref{m}, \eqref{a}, \eqref{b}, \eqref{c^kappa,nu} and \eqref{Imaginary conditions}).
By Lemma \ref{l:Phi_0} we get
\begin{align}\label{Phi_0 abs bound}
 \big\|\Phi_{0, \theta}^{\nu, \kappa}(\lambda; r)\big\|_{\mathbb C_2}^2 &\leqslant 2\cos^2\!\theta\,\lambda^{2\beta}\big\|\Phi_U^{\nu, \kappa}(\lambda r)\big\|_{\mathbb C_2}^2 +2\sin^2\!\theta\,\lambda^{-2\beta}\big\|\Phi_M^{\nu, \kappa}(\lambda r)\big\|_{\mathbb C_2}^2.
\end{align}
By Lemmata \ref{Lemma solutions dirac equation} and \ref{lemma asymptotics at zero} for $\beta >0$ there exist finite constants $C_M^{\nu, \kappa}$ and $C_U^{\nu, \kappa}$ such that for any $x\in \mathbb R_+$
\begin{align}\label{Phi_M and Phi_U estimates}\begin{split}
 \big\|\Phi_M^{\nu, \kappa}(x)\big\|_{\mathbb C_2}^2 \leqslant C_M^{\nu, \kappa}\begin{cases}x^{2\beta}, &\text{\!\!\!for }x\leqslant 1\\ 1, &\text{\!\!\!for }x\geqslant 1\end{cases}; \  \big\|\Phi_U^{\nu, \kappa}(x)\big\|_{\mathbb C_2}^2 \leqslant C_U^{\nu, \kappa}\begin{cases}x^{-2\beta}, &\text{\!\!\!for }x\leqslant 1\\ 1, &\text{\!\!\!for }x\geqslant 1\end{cases}.
\end{split}\end{align}

Substituting \eqref{Phi_0 abs bound} into \eqref{integral estimate first use} and using the estimates \eqref{Phi_M and Phi_U estimates} and \eqref{tilde m bound} we obtain \eqref{integral estimate for beta>0}.

For $\beta> 0$ and $\theta =\pi/2$, starting from \eqref{integral estimate applied}, rescaling $\lambda =:\mu/r$ and using Lemma \ref{l:Phi_0} and \eqref{m} we conclude 
\begin{align*}\begin{split}
 &\rank P_{(-\infty, 0)}\big(D_{\theta, \infty}^{\nu, \kappa}(V)\big)\\ &\leqslant m_{\pi/2}^{\nu, \kappa}(1)\int_0^\infty \mu^{-q}\big\|\Phi_M^{\nu, \kappa}(\mu)\big\|_{\mathbb C^2}^2\mathrm d\mu\int_0^\infty \big\|V_+^{q}(r)\big\|_{\mathbb C^{2\times 2}}r^{q -1}\,\mathrm dr,\end{split}
\end{align*}
where the $\mu$-integral is finite by \eqref{Phi_M and Phi_U estimates}.
\end{proof}

We now turn to the case of $\beta =0$, in which the integral on the right hand side of \eqref{rank via integral} does not converge for $E =\infty$. To remedy this problem, we use
\begin{lem}\label{l: energy splitting}
Let $\|V_+\|_\infty :=\|V_+\|_{\mathsf L^\infty(\mathbb R_+, \mathbb C^{2\times 2})}$ be finite. Then for any $\tau \geqslant 0$ we have
\begin{align*}
 \rank P_{(-\infty, -\tau)}\big(D_{\theta, \infty}^{\nu, \kappa}(V)\big) \leqslant \rank P_{(-\infty, -\tau)}\Big(D_{\theta, 2\|V_+\|_\infty}^{\nu, \kappa}\big(2V_+\big)\Big).
\end{align*}
\end{lem}

\begin{proof}
 For any $\varphi\in P_{\theta, \infty}^{\nu, \kappa}\mathfrak D\big(|D^{\nu, \kappa}_\theta|^{1/2}\big)$ let $\check\varphi :=P_{\theta, 2\|V_+\|_\infty}^{\nu, \kappa}\varphi$; $\hat\varphi :=\varphi -\check\varphi$. Then \eqref{d form} implies
\begin{align}\label{form V splitting}\begin{split}
 &\mathfrak d_{\theta, \infty}^{\nu, \kappa}(V)[\varphi] \geqslant\mathfrak d_{\theta, \infty}^{\nu, \kappa}(V_+)[\check\varphi] +\mathfrak d_{\theta, \infty}^{\nu, \kappa}(V_+)[\hat\varphi] -2\Re\langle V_+^{1/2}\check\varphi, V_+^{1/2}\hat\varphi\rangle_{\mathsf L^2(\mathbb R_+, \mathbb C^2)}\\ &\geqslant \mathfrak d_{\theta, 2\|V_+\|_\infty}^{\nu, \kappa}\big(2V_+\big)[\check\varphi] +\Big(\big\||D^{\nu, \kappa}_\theta|^{1/2}\hat\varphi\big\|_{\mathsf L^2(\mathbb R_+, \mathbb C^2)}^2 -2\big\|V_+^{1/2}\hat\varphi\big\|_{\mathsf L^2(\mathbb R_+, \mathbb C^2)}^2\Big).
\end{split}\end{align}
Since $\big\||D^{\nu, \kappa}_\theta|^{1/2}\hat\varphi\big\|_{\mathsf L^2(\mathbb R_+, \mathbb C^2)}^2 \geqslant 2\|V_+\|_\infty\|\hat\varphi\|_{\mathsf L^2(\mathbb R_+, \mathbb C^2)}^2$, the term in the parenthesis on the right hand side of \eqref{form V splitting} is non-negative. The statement of the lemma now follows from the minimax principle.
\end{proof}

\begin{proof}[Proof of Theorem \ref{t: the alternative}, part \emph{III}]
 Combining Lemma \ref{l: energy splitting} with Lemma \ref{l:IntegralEstimate} (with $\tau \to +0$, $q :=1$) we obtain
\begin{align}\begin{split}\label{BS for beta=0}
 &\rank P_{(-\infty, 0)}\big(D_{\theta, \infty}^{\nu, \kappa}(V)\big)\\ &\leqslant \int_0^\infty \big\|V_+(r)\big\|_{\mathbb C^{2\times 2}}\int_0^{2\|V_+\|_\infty} \lambda^{-1}m^{\nu, \kappa}_\theta(\lambda)\big\|\Phi_{0, \theta}^{\nu, \kappa}(\lambda; r)\big\|_{\mathbb C^2}^2\mathrm d\lambda\,\mathrm dr.
\end{split}\end{align}
Inserting the definitions \eqref{m}, \eqref{Phi_0}, \eqref{a} and \eqref{b} and performing the rescaling $\lambda =:\mathrm e^{\tan\theta}\mu$ we rewrite the inner integral on the right hand side of \eqref{BS for beta=0} as
\begin{align}\begin{split}\label{integral via f and g}
 &\int_0^{2\|V_+\|_\infty} \lambda^{-1}m^{\nu, \kappa}_\theta(\lambda)\big\|\Phi_{0, \theta}^{\nu, \kappa}(\lambda; r)\big\|_{\mathbb C^2}^2\mathrm d\lambda \\&= \int_0^{2\|V_+\|_\infty\mathrm e^{-\tan\theta}} \mu^{-1}m_0^{\nu, \kappa}(\mu)\big\|f^{\nu, \kappa}(\mu\mathrm e^{\tan\theta} r) +g^{\nu, \kappa}(\mu; \mathrm e^{\tan\theta}r)\big\|_{\mathbb C^2}^2\mathrm d\mu 
\end{split}\end{align}
with $f^{\nu, \kappa}:\mathbb R_+\to \mathbb C^2$ and $g^{\nu, \kappa}:\mathbb R_+^2\to \mathbb C^2$ given by
\begin{align*}
 f^{\nu, \kappa}(x) :=\Phi_U^{\nu, \kappa}(x) -\ln(x)\Phi_M^{\nu, \kappa}(x), \quad g^{\nu, \kappa}(\zeta, x) :=\ln(x)\Phi_M^{\nu, \kappa}(\zeta x).
\end{align*}
Lemma \ref{lemma asymptotics at zero} and boundedness of $\Phi_U^{\nu, \kappa}$ on $[1, \infty)$ and $\Phi_M^{\nu, \kappa}$ on $\mathbb R_+$ imply the estimates
\begin{align*}%\label{f and g estimates}
 \big\|f^{\nu, \kappa}(\cdot)\big\|_{\mathbb C^2}^2 \leqslant C_f^{\nu, \kappa}\ln^2(\mathrm e +\cdot), \quad \big\|g^{\nu, \kappa}(\zeta, \cdot)\big\|_{\mathbb C^2}^2 \leqslant C^{\nu, \kappa}_g\ln^2(\cdot)
\end{align*}
with finite constants $C_f^{\nu, \kappa}$ and $C^{\nu, \kappa}_g$ independent of $\zeta$. Moreover, \eqref{m} and \eqref{Imaginary conditions 2} imply the bound
\begin{align*}%\label{m_0 bound}
 m_0^{\nu, \kappa}(\cdot) \leqslant C_m^{\nu, \kappa}\big(1 +\ln^2(\cdot)\big)^{-1}.
\end{align*}
Applying the above estimates to the right hand side of \eqref{integral via f and g} and using the monotonicity of the logarithm we conclude
\begin{align*}
 &\int_0^{2\|V_+\|_\infty} \lambda^{-1}m^{\nu, \kappa}_\theta(\lambda)\big\|\Phi_{0, \theta}^{\nu, \kappa}(\lambda; r)\big\|_{\mathbb C^2}^2\mathrm d\lambda\\ &\leqslant 2C_m^{\nu, \kappa}\Big(C_f^{\nu, \kappa}\ln^2\big(\mathrm e +2\|V_+\|_\infty r\big) +C^{\nu, \kappa}_g\ln^2(\mathrm e^{\tan\theta}r)\Big)\int_0^{2\|V_+\|_\infty\mathrm e^{-\tan\theta}}\hspace{-0.8cm}\frac{\mathrm d\mu}{\mu(1 +\ln^2\mu)},
\end{align*}
where the integral can be estimated by $\pi$. Substituting into \eqref{BS for beta=0} we obtain \eqref{no VL for beta=0}.
\end{proof}

\section{Proof of Theorem \ref{t: LT}}

\begin{lem}\label{l: number of eigenvalues below tau}
For $\tau >0$, $q >1$, and the combinations of $\nu, \kappa\in \mathbb R$ and $\theta \in[0, \pi)$ covered in \eqref{weights} there exists a constant $k^{\nu, \kappa}> 0$ such that
\begin{align*}
 \rank P_{(-\infty, -\tau)}\big(D_{\theta, \infty}^{\nu, \kappa}(V)\big) \leqslant \frac{\tau^{1 -q}}{q -1}k^{\nu, \kappa}\int_0^\infty \big\|V_+(r)\big\|^q_{\mathbb C^{2\times2}}W^{\nu, \kappa}_\theta(r)\,\mathrm dr
\end{align*}
holds.
\end{lem}

\begin{proof}
Lemma \ref{l:IntegralEstimate} implies
\begin{align*}\begin{split}
 &\rank P_{(-\infty, -\tau)}\big(D_{\theta, \infty}^{\nu, \kappa}(V)\big)\\ &\leqslant \frac{\tau^{1 -q}}{q -1}\int_0^\infty \!\!\!\big\|V_+(r)\big\|^q_{\mathbb C^{2\times 2}}\Big\|\sqrt{m^{\nu, \kappa}_\theta(\cdot)}\Phi_{0, \theta}^{\nu, \kappa}(\cdot; r)\Big\|_{\mathsf L^\infty(\mathbb R_+, \mathbb C^2)}^2\mathrm dr\end{split}
\end{align*}
holds for all $q >1$. It remains to obtain the estimate
\begin{align}\label{weight control}
 \Big\|\sqrt{m^{\nu, \kappa}_\theta(\cdot)}\Phi_{0, \theta}^{\nu, \kappa}(\cdot; r)\Big\|_{\mathsf L^\infty(\mathbb R_+, \mathbb C^2)}^2 \leqslant k^{\nu, \kappa}W^{\nu, \kappa}_{\theta}(r)
\end{align}
for all $r >0$ with finite $k^{\nu, \kappa}> 0$ and $W^{\nu, \kappa}_{\theta}$ defined in \eqref{weights}.

\emph{Case $\nu^2 = \kappa^2\neq 0$ and $\theta =\pi/2$\emph:} By \eqref{m} and Lemma \ref{l:Phi_0} $m^{\nu, \kappa}_{\pi/2}(\lambda) =m^{\nu, \kappa}_{\pi/2}(1)$ for all $\lambda \in\mathbb R_+$ and for all $r >0$ we have 
\begin{align*}
\big\|\Phi_{0, \pi/2}^{\nu, \kappa}(\cdot; r)\big\|_{\mathsf L^\infty(\mathbb R_+, \mathbb C^2)} =\big\|\Phi_M^{\nu, \kappa}(\cdot\,r)\big\|_{\mathsf L^\infty(\mathbb R_+, \mathbb C^2)} = \big\|\Phi_M^{\nu, \kappa}\big\|_{\mathsf L^\infty(\mathbb R_+, \mathbb C^2)} < \infty. 
\end{align*}

\emph{Case $\kappa \neq 0$, $\beta \in \mathrm i\mathbb R_+$ and $\theta \in[0, \pi)$\emph:} A computation based on Formula 5.4.3 in \cite{dlmf} allows us to justify the relations
\begin{align*}%\label{c algebraic relation}
 c_{+, +}^{\nu, \kappa}\cdot\overline{c_{-}^{\nu, \kappa}} =1\quad \text{ and }\quad c_{+, +}^{\nu, \kappa}({c_{-}^{\nu, \kappa}})^{-1} \gtrless 1\text{ for } \nu \gtrless 0.
\end{align*}
for the coefficients introduced in \eqref{c_+,pm}, \eqref{c_-} provided $\beta \in\mathbb R_+$.
Hence there exist $\rho\in (1, \infty)$ and $\omega\in [0, 2\pi)$ depending only on $\kappa$ and $\nu$ such that
\begin{align*}
 c_{+, +}^{\nu, \kappa} =\rho^{\sign\nu}\mathrm e^{\mathrm i\omega}\quad \text{and}\quad c_{-}^{\nu, \kappa} =\rho^{-\sign\nu}\mathrm e^{\mathrm i\omega}.
\end{align*}
Substituting this into \eqref{m} by \eqref{c^kappa,nu}, \eqref{W_MU}, \eqref{a} and \eqref{b} we get
\begin{align*}
 m_\theta^{\nu, \kappa}(\lambda) &=\frac1{4\pi\kappa^2|\nu||\beta|}\cdot\frac{\rho -\rho^{-1}}{\rho +\rho^{-1} -2\cos\big(\omega +2\theta +2|\beta|\ln\lambda\big)}\\ &\leqslant\frac1{4\pi\kappa^2|\beta\nu|}\cdot\frac{\rho -\rho^{-1}}{\rho+ \rho^{-1} -2} =\frac1{4\pi\kappa^2|\beta\nu|}\cdot\frac{\rho +1}{\rho -1}.
\end{align*}
Taking into account
\begin{align*}
 \big\|\Phi_{0, \theta}^{\nu, \kappa}(\cdot\,; r)\big\|_{\mathsf L^\infty(\mathbb R_+, \mathbb C^2)} &=\Big\|2\Re\big((\cdot)^{-\beta}\mathrm e^{-\mathrm i\theta}\Phi_M^{\nu, \kappa}(\cdot\, r)\big)\Big\|_{\mathsf L^\infty(\mathbb R_+, \mathbb C^2)}\\ &\leqslant 2\big\|\Phi_M^{\nu, \kappa}\big\|_{\mathsf L^\infty(\mathbb R_+, \mathbb C^2)} <\infty
\end{align*}
we obtain \eqref{weight control}.

\emph{Case $\kappa =0$, $\nu\in\mathbb R$ and $\theta\in[0, \pi)$\emph:} Here \eqref{weight control} follows immediately from \eqref{m^0}, Lemma \ref{l:Phi_0} and \eqref{Phi for kappa=0}.

\emph{Case $\nu^2 = \kappa^2\neq 0$ and $\theta \in [0, \pi)\setminus\{\pi/2\}$\emph:} The left hand side of \eqref{weight control} coincides with $Z^{\nu, \kappa}_0(\mathrm e^{\tan\theta}r)$, where $Z^{\nu, \kappa}_0:\mathbb R_+\to \mathbb R$ is defined by
\begin{align}\label{W_0}
 Z^{\nu, \kappa}_0(x) :=\sup_{\zeta \in\mathbb R_+}\frac{|c_- -c_{+, +}|\big\|\Phi_U^{\nu, \kappa}(\zeta x) -(\ln\zeta)\Phi_M^{\nu, \kappa}(\zeta x)\big\|_{\mathbb C^2}^2}{\big|2\pi\nu(c_{+, +}^{\nu, \kappa} +\ln\zeta)(c_{-}^{\nu, \kappa} +\ln\zeta)\big|}.
\end{align}
By Lemma \ref{l:no poles} we can estimate
\begin{align}\label{denominator of W_0}
 \frac{|c_- -c_{+, +}|}{\big|2\pi\nu(c_{+, +}^{\nu, \kappa} +\ln\zeta)(c_{-}^{\nu, \kappa} +\ln\zeta)\big|} \leqslant \frac{K^{\nu, \kappa}}{1 +\ln^2\zeta}
\end{align}
with some $K^{\nu, \kappa}\in \mathbb R_+$. On the other hand, Lemma \ref{lemma asymptotics at zero} and boundedness of $\Phi_U^{\nu, \kappa}$ and $\Phi_M^{\nu, \kappa}$ on $(1, \infty)$ imply
\begin{align}\label{numerator of W_0}
 \big\|\Phi_U^{\nu, \kappa}(\zeta x) -(\ln\zeta)\Phi_M^{\nu, \kappa}(\zeta x)\big\|_{\mathbb C^2}^2 \leqslant L^{\nu, \kappa}\begin{cases}
                                                                                                                              1 +\ln^2x, &\text{for }\zeta \leqslant x^{-1};\\ 1 +\ln^2\zeta, &\text{for }\zeta \geqslant x^{-1}
                                                                                                                             \end{cases}
\end{align}
with finite $L^{\nu, \kappa}$. Inserting \eqref{denominator of W_0} and \eqref{numerator of W_0} into \eqref{W_0} we obtain \eqref{weight control}.

\emph{Case $\beta \in (0, 1/2)$ and $\theta \in (0, \pi)$ or $\beta \geqslant1/2$ and $\theta =\pi/2$\emph:} For $\theta =\pi/2$ the left hand side of \eqref{weight control} is given by $m_{\pi/2}^{\nu, \kappa}(1)\|\Phi_M^{\nu, \kappa}\|_{\mathsf L^\infty(\mathbb R_+, \mathbb C^2)}^2$ and the statement follows. Otherwise, by \eqref{m}, \eqref{W_MU} and Lemma \ref{l:Phi_0} we have
\begin{align}\label{technical weight}
 \Big\|\sqrt{m^{\nu, \kappa}_\theta(\cdot)}\Phi_{0, \theta}^{\nu, \kappa}(\cdot; r)\Big\|_{\mathsf L^\infty(\mathbb R_+, \mathbb C^2)}^2 =Z^{\nu, \kappa}_{\sign(\tan\theta)}\big(|\tan\theta|^{1/(2\beta)}r\big),
\end{align}
where $Z^{\nu, \kappa}_\pm: \mathbb R_+\to \mathbb R$ is defined by
\begin{align}\label{W_pm}
 Z^{\nu, \kappa}_\pm(x) :=\sup_{\zeta \in\mathbb R_+}\frac{|c_- -c_{+, +}|\big\|\zeta\Phi_U^{\nu, \kappa}(\zeta^{1/(2\beta)} x) \pm\Phi_M^{\nu, \kappa}(\zeta^{1/(2\beta)} x)\big\|_{\mathbb C^2}^2}{\big|2\pi W[\Phi_M^{\nu, \kappa}, \Phi_U^{\nu, \kappa}](c_{+, +}^{\nu, \kappa}\zeta \mp1)(c_{-}^{\nu, \kappa}\zeta \mp1)\big|}.
\end{align}
By \eqref{W_MU} and \eqref{Imaginary conditions} we can estimate
\begin{align}\label{denominator of W_pm}
 \frac{|c_- -c_{+, +}|}{\big|2\pi W[\Phi_M^{\nu, \kappa}, \Phi_U^{\nu, \kappa}](c_{+, +}^{\nu, \kappa}\zeta \mp1)(c_{-}^{\nu, \kappa}\zeta \mp1)\big|} \leqslant \frac{X^{\nu, \kappa}}{1 +\zeta^2}
\end{align}
with some finite $X^{\nu, \kappa}$. On the other hand, \eqref{Phi_M and Phi_U estimates} implies
\begin{align}\label{numerator of W_pm}\begin{split}
 &\big\|\zeta\Phi_U^{\nu, \kappa}(\zeta^{1/(2\beta)} x) \pm\Phi_M^{\nu, \kappa}(\zeta^{1/(2\beta)} x)\big\|_{\mathbb C^2}^2\\& \leqslant Y^{\nu, \kappa}\begin{cases}
                                                                                                                              \zeta(x^{-2\beta} +x^{2\beta}), &\text{for }\zeta \leqslant x^{-2\beta};\\ 1 +\zeta^2, &\text{for }\zeta \geqslant x^{-2\beta}
                                                                                                                             \end{cases}
\end{split}\end{align}
with finite $Y^{\nu, \kappa}$. Inserting \eqref{denominator of W_pm} and \eqref{numerator of W_pm} into \eqref{W_pm} we obtain
\begin{align*}
 Z^{\nu, \kappa}_\pm(x) &\leqslant X^{\nu, \kappa}Y^{\nu, \kappa}\max\Big\{1, (x^{-2\beta} +x^{2\beta})\sup_{\zeta \leqslant x^{-2\beta}}\frac\zeta{1 +\zeta^2}\Big\}\\ &\leqslant X^{\nu, \kappa}Y^{\nu, \kappa}\max\{1, x^{-2\beta}\}.
\end{align*}
Inserting into \eqref{technical weight} we get \eqref{weight control}.
\end{proof}

\begin{proof}[Proof of Theorem \ref{t: LT}, part a).]
 The statement follows from Lemma \ref{l: number of eigenvalues below tau} by the\\ usual arguments (see e.g. the proof of Theorem 1.1 in \cite{Frank2009}): Let $\gamma >0$ and $q \in(1, 1+ \gamma)$. By the minimax principle we have 
\begin{align}\label{HLT via CLR}\begin{split}
 &\tr\big(D_{\theta, \infty}^{\nu, \kappa}(V)\big)_-^\gamma =\int_0^\infty\gamma\tau^{\gamma -1}\rank P_{(-\infty, -\tau)}\big(D_{\theta, \infty}^{\nu, \kappa}(V)\big)\mathrm d\tau\\ &\leqslant \int_0^\infty\gamma\tau^{\gamma -1}\rank P_{(-\infty, -\tau/2)}\Big(D_{\theta, \infty}^{\nu, \kappa}\big((V_+ -\tau/2)_+\big)\Big)\mathrm d\tau\\ &\leqslant \frac{\gamma k^{\nu, \kappa}}{q -1}\int_0^\infty W_\theta^{\nu, \kappa}(r)\int_0^\infty\tau^{\gamma -1}\big\|(V_+(r) -\tau/2)_+\big\|^q_{\mathbb C^{2\times2}}(\tau/2)^{1 -q}\mathrm d\tau\,\mathrm dr.
\end{split}\end{align}
Passing to the new integration variable $\sigma :=\tau\big\|V_+(r)\big\|_{\mathbb C^{2\times2}}^{-1}/2$ in the inner integral we can rewrite the right hand side of \eqref{HLT via CLR} as
\begin{align*}
2^\gamma\gamma k^{\nu, \kappa}\bigg(\frac1{q -1}\int_0^1(1 -\sigma)^q\sigma^{\gamma -q}\mathrm d\sigma\bigg)\int_0^\infty\big\|V_+(r)\big\|_{\mathbb C^{2\times2}}^{1 +\gamma}W^{\nu, \kappa}_{\theta}(r)\,\mathrm dr.
\end{align*}
Minimising in $q \in(1, 1+ \gamma)$ we arrive at \eqref{LT}.
\end{proof}

\begin{proof}[Proof of Theorem \ref{t: LT}, part b)]
Let $q \in(1, 1 +\gamma -2\beta)$. Applying Lemma \ref{l:IntegralEstimate} and the minimax principle we get by \eqref{m} and Lemma \ref{l:Phi_0}
\begin{align}\begin{split}\label{estimate with Omega}
  &\tr\big(D_{0, \infty}^{\nu, \kappa}(V)\big)_-^\gamma =\int_0^\infty\gamma\tau^{\gamma -1}\rank P_{(-\infty, -\tau)}\big(D_{0, \infty}^{\nu, \kappa}(V)\big)\mathrm d\tau\\ &\leqslant \int_0^\infty\gamma\tau^{\gamma -1}\rank P_{(-\infty, -\tau/2)}\Big(D_{0, \infty}^{\nu, \kappa}\big((V_+ -\tau/2)_+\big)\Big)\mathrm d\tau\\ &\leqslant\int_0^\infty\!\!\int_0^{2\|V_+(r)\|_{\mathbb C^{2\times2}}}\!\!\!\gamma\tau^{\gamma -1}m_0^{\nu, \kappa}(1)\Big(\big\|V_+(r)\big\|_{\mathbb C^{2\times2}} -\frac\tau2\Big)^q\Big(\frac\tau2\Big)^{1 -q}\Omega_q^{\nu, \kappa}\Big(\frac{\tau r}2\Big)\mathrm d\tau\,\mathrm dr,
\end{split}\end{align}
with $\Omega_q^{\nu, \kappa}:\mathbb R_+ \to\mathbb R$ given by
\begin{align*}
 \Omega_q^{\nu, \kappa}(x) :=\int_0^\infty\frac{\big\|\Phi_U^{\nu, \kappa}(\lambda x)\big\|_{\mathbb C^{2}}^2}{(1 +\lambda)^q}\mathrm d\lambda.
\end{align*}
Employing \eqref{Phi_M and Phi_U estimates} we conclude the existence of $H^{\nu, \kappa}_q \in\mathbb R_+$ such that
\begin{align*}
\Omega_q^{\nu, \kappa} \leqslant H^{\nu, \kappa}_q\big((\cdot)^{-2\beta} +1\big). 
\end{align*}
Substituting into \eqref{estimate with Omega}, computing the inner integral and minimising over\\ $q \in(1, 1 +\gamma -2\beta)$ we obtain \eqref{LT modified}.
\end{proof}

\section{Proof of Theorem \ref{t: virtual levels 2d}}

We will need the following lemma to control higher channels in the angular momentum decomposition.

\begin{lem}\label{l: higher channels}
Let $\nu\in \mathbb R$. For all 
\begin{align}\label{kappa is big}
\kappa\in \mathbb Z+1/2\quad \text{with}\quad |\kappa| \geqslant \kappa_\nu :=\min\{\kappa\in \mathbb N +1/2: \kappa^2 >\nu^2 +1/4\}
\end{align}
the estimate
\begin{align}\label{lower bound on modulus}
 |D^{\nu, \kappa}_{\pi/2}|\geqslant C^\nu |D^{0, \kappa}_{\pi/2}|
\end{align}
holds with
\begin{align*}%\label{high momentum constant}
 C^\nu :=1 -\frac{\nu^2(\kappa_\nu^2 +1/4)}{(\kappa_\nu^2 -1/4)^2}\bigg(\sqrt{1 +\frac{(4\kappa_\nu^2 -\nu^2)(\kappa_\nu^2 -1/4)^2}{(\kappa_\nu^2 +1/4)^2\nu^2}} -1\bigg) >0.
\end{align*}
\end{lem}

\begin{proof}
The case of $\nu =0$ is trivial. Otherwise we proceed as in the proof of Lemma 28 in \cite{MorozovMueller}. The only difference is that we now have relaxed the assumptions on $\kappa$ and $\nu$. 

As explained in the original proof, the statement of the lemma holds true provided for $b:= 1 -C^\nu$ we have $b <1$ and the inequality 
\begin{align}\label{a nu}
a^\nu_{\kappa, -}(b, s) :=\nu^2 +b/4 +\kappa^2b +s^2b -(4\kappa^2\nu^2 +4\nu^2s^2 +\kappa^2b^2)^{1/2} \geqslant 0
\end{align}
holds for all $\kappa$ satisfying \eqref{kappa is big} and $s\in\mathbb R$.
Since \eqref{a nu} does not depend on the signs of $\kappa$ and $s$, without loss of generality we can assume $\kappa, s\in \overline{\mathbb R_+}$.
Extending \eqref{a nu} to $\kappa \in\mathbb R$, we get for $\kappa >0$ satisfying \eqref{kappa is big}
\begin{align*}
 &a^\nu_{\kappa +1, -}(b, s) -a^\nu_{\kappa, -}(b, s) =\int_{\kappa}^{\kappa +1}\frac{\partial a^\nu_{\varkappa, -}}{\partial \varkappa}(b, s)\,\mathrm d\varkappa\\
&=(2\kappa +1)b -\int_{\kappa}^{\kappa +1}\frac{(4\nu^2 +b^2)\varkappa}{\sqrt{(4\nu^2 +b^2)\varkappa^2 +4\nu^2s^2}}\,\mathrm d\varkappa \geqslant (2\kappa_\nu +1)b -\sqrt{4\nu^2 +b^2},
\end{align*}
so for all other parameters being fixed, $a^\nu_{\kappa, -}(b, s)$ is an increasing function of $\kappa \geqslant \kappa_\nu$ provided 
\begin{align*}
 (2\kappa_\nu +1)b -\sqrt{4\nu^2 +b^2} \geqslant 0,
\end{align*}
i.e.
\begin{align*}
 b \geqslant |\nu|/\sqrt{\kappa_\nu^2 +\kappa_\nu}
\end{align*}
holds.
For $s >0$ we have
\begin{align*}
 \frac{\partial a^\nu_{\kappa_\nu, -}(b, s)}{\partial s} &= 2s\Big(b -\frac{2\nu^2}{\sqrt{4\kappa_\nu^2\nu^2 +4\nu^2s^2 +\kappa_\nu^2b^2}}\Big) \geqslant 2s\Big(b -\frac{2\nu^2}{\kappa_\nu\sqrt{4\nu^2 +b^2}}\Big).
\end{align*}
Thus, provided 
\begin{align}\label{b preliminary estimate}
 b \geqslant\sqrt2|\nu|\big(\sqrt{1 +\kappa_\nu^{-2}} -1\big)^{1/2}\qquad \Big(> |\nu|/\sqrt{\kappa_\nu^2 +\kappa_\nu}\Big)
\end{align}
holds, for any $s\in \mathbb R_+$ we have $(\partial a^\nu_{\kappa_\nu, -}/\partial s)(b, s) \geqslant 0$ and hence for $\kappa \geqslant \kappa_\nu$ 
\begin{align}\label{a kappa_nu estimate}
 a^\nu_{\kappa, -}(b, s)\geqslant a^\nu_{\kappa_\nu, -}(b, 0) =\nu^2 +b/4 +\kappa_\nu^2b -\kappa_\nu(4\nu^2 +b^2)^{1/2}
\end{align}
holds. The right hand side of \eqref{a kappa_nu estimate} is positive provided $b >0$ satisfies
\begin{align}\label{b quadratic inequality}
 f_\nu(b) :=(\kappa_\nu^2 -1/4)^2b^2 +2(\kappa_\nu^2 +1/4)\nu^2b +\nu^4 -4\nu^2\kappa_\nu^2 \geqslant 0.
\end{align}
Here $f_\nu$ is a quadratic function with $f_\nu(0) <0$ and $f_\nu(1) =(\kappa_\nu^2 -\nu^2 -1/4)^2 >0$. Thus $f_\nu$ has a unique zero in $(0, 1)$ which coincides with $1 -C^\nu$. Hence for $b >0$ \eqref{b quadratic inequality} is equivalent to $b \geqslant 1 -C^\nu$. Note that this condition is more restrictive than \eqref{b preliminary estimate}, since 
\begin{align*}
 \sqrt2|\nu|\big(\sqrt{1 +\kappa_\nu^{-2}} -1\big)^{1/2} \leqslant |\nu|/\kappa_\nu
\end{align*}
and by \eqref{kappa is big}
\begin{align*}
  f_\nu\big(|\nu|/\kappa_\nu\big) \leqslant \nu^2(\kappa_\nu^{-2}/16 +\nu^2 -\kappa_\nu^2) \leqslant \nu^2(1/4 +\nu^2 -\kappa_\nu^2) <0
\end{align*}
holds.
\end{proof}

\begin{proof}[Proof of Theorem \ref{t: virtual levels 2d}, part 1.] 
For any $\psi\in P_{[0, \infty)}(D^{\nu, \kappa_0}_{\boldsymbol\theta(\kappa_0)})\mathfrak D(D^{\nu, \kappa_0}_{\boldsymbol\theta(\kappa_0)})$ we have
\begin{align*}
\Psi :=\mathcal A^* \,\bigoplus_{\kappa\in \mathbb Z+1/2}\delta_{\kappa, \kappa_0}\psi \in P_{[0, \infty)}(D^{\nu}_{\boldsymbol\theta})\mathfrak D(D^{\nu}_{\boldsymbol\theta})
\end{align*}
(here $\delta_{\kappa, \kappa_0}$ is a Kronecker delta) and by \eqref{A} and \eqref{Coulomb-Dirac 2D}
\begin{align}\label{channel passage}
 \langle\Psi, D^{\nu}_{\boldsymbol\theta}(Q)\Psi\rangle_{P_{[0, \infty)}(D^{\nu}_{\boldsymbol\theta})\mathsf L^2(\mathbb R^2, \mathbb C^2)} =\langle\psi, D^{\nu, \kappa_0}_{\boldsymbol\theta(\kappa_0)}(V)\psi\rangle_{P_{[0, \infty)}(D^{\nu, \kappa_0}_{\boldsymbol\theta(\kappa_0)})\mathsf L^2(\mathbb R_+, \mathbb C^2)}
\end{align}
holds with $V$ as defined in \eqref{V}.
But since all the assumptions of Theorem \ref{t: the alternative}, part I are satisfied for $D^{\nu, \kappa_0}_{\boldsymbol\theta(\kappa_0)}(V)$, we can choose $\psi$ so that \eqref{channel passage} is negative. Hence by the minimax principle $D^{\nu}_{\boldsymbol\theta}(Q)$ has non-empty negative spectrum.
\end{proof}

\begin{proof}[Proof of Theorem \ref{t: virtual levels 2d}, part 2.]
Given $Q$ satisfying \eqref{R bounds Q} and $\alpha\in \mathbb R_+$, we have
\begin{align}\label{orthogonal sum}
 D^{\nu}_{\boldsymbol\theta}(\alpha Q) \geqslant D^{\nu}_{\boldsymbol\theta}\Big(\alpha R\big(|\cdot|\big)\mathbb I\Big) =\mathcal A^*\bigg(\bigoplus_{\kappa\in \mathbb Z+1/2}D^{\nu, \kappa}_{\boldsymbol\theta(\kappa)}(\alpha R\mathbb I)\bigg)\mathcal A.
\end{align}
It is thus enough to show that every term on the right hand side of \eqref{orthogonal sum} has no negative spectrum for $\alpha\in [0, \alpha_c)$ with $\alpha_c >0$ independent of $\kappa$.

For all $\kappa\in \mathbb Z+1/2$ with $\kappa^2\leqslant \nu^2 +1/4$ assumption (a) or (b) and Theorem \ref{t: the alternative}, part II or III, respectively, imply the existence of $\alpha_\kappa> 0$ such that
\begin{align}\label{D is positive}
D^{\nu, \kappa}_{\boldsymbol\theta(\kappa)}(\alpha R\mathbb I) \geqslant 0 
\end{align}
holds for $\alpha\in [0, \alpha_\kappa)$.

Let us now consider $\kappa$ satisfying \eqref{kappa is big}.

Assumption (c) of the theorem means that there exist $R_1\in \mathsf L^\infty(\mathbb R_+, r\mathrm dr)$ and $R_2\in \mathsf L^2(\mathbb R_+, r\mathrm dr)$ such that $R =R_1 +R_2$.
We have
\begin{align}\label{algebraic doubling}
D^{\nu, \kappa}_{\boldsymbol\theta(\kappa)}(\alpha R\mathbb I) =\frac12\big(D^{\nu, \kappa}_{\boldsymbol\theta(\kappa)}(2\alpha R_1\mathbb I) +D^{\nu, \kappa}_{\boldsymbol\theta(\kappa)}(2\alpha R_2\mathbb I)\big).
\end{align}

Theorem 2.5 in \cite{Herbst1977} implies
\begin{align*}
\mathcal A^*\bigg(\bigoplus_{\kappa\in \mathbb Z+1/2}|D^{\kappa, 0}_{\pi/2}|\bigg)\mathcal A = (-\Delta)^{1/2}\mathbb I \geqslant K|\cdot|^{-1}\mathbb I =\mathcal A^*\bigg(\bigoplus_{\kappa\in \mathbb Z+1/2}K/(\cdot)\mathbb I\bigg)\mathcal A
\end{align*}
with
\begin{align*}
 K:= 2\frac{\Gamma^2(3/4)}{\Gamma^2(1/4)},
\end{align*}
from which we conclude
\begin{align}\label{Kato 1D}
 |D^{\kappa, 0}_{\pi/2}| \geqslant K/(\cdot), \qquad\text{for all }\kappa\in \mathbb Z+1/2.
\end{align}
For all $\kappa$ satisfying \eqref{kappa is big} by Theorem \ref{t: self-adjoint realisations} we have $\boldsymbol\theta(\kappa) =\pi/2$. Combining this with \eqref{lower bound on modulus} and \eqref{Kato 1D} we obtain
\begin{align*}%\label{lower bound combined}
 |D^{\nu, \kappa}_{\boldsymbol\theta(\kappa)}|\geqslant C^\nu K/(\cdot)
\end{align*}
for all $\kappa$ satisfying \eqref{kappa is big}. Since $R_1\in \mathsf L^\infty(\mathbb R_+, r\mathrm dr)$, there exists $\alpha_0> 0$ such that for all $\kappa$ satisfying \eqref{kappa is big} 
\begin{align}\label{R_1 positivity}
D^{\nu, \kappa}_{\boldsymbol\theta(\kappa)}(2\alpha R_1\mathbb I) \geqslant 0 
\end{align}
holds for $\alpha\in [0, \alpha_0)$.

For $R_2$ we use \eqref{lower bound on modulus} to obtain
\begin{align}\begin{split}\label{R_2 game}
 &\rank P_{(-\infty, 0)}\big(D^{\nu, \kappa}_{\boldsymbol\theta(\kappa)}(2\alpha R_2\mathbb I)\big) \leqslant \rank P_{(-\infty, 0)}\big(|D^{\kappa, 0}_{\pi/2}| -2(C^\nu)^{-1}\alpha R_2\mathbb I\big)\\ &\leqslant \sum_{\kappa\in \mathbb Z +1/2}\rank P_{(-\infty, 0)}\big(|D^{0, \kappa}_{\pi/2}| -2(C^\nu)^{-1}\alpha R_2\mathbb I\big)\\ &=\rank P_{(-\infty, 0)}\Big((-\Delta)^{1/2}\mathbb I -2(C^\nu)^{-1}\alpha R_2\big(|\cdot|\big)\mathbb I\Big)\\ &=2\rank P_{(-\infty, 0)}\Big((-\Delta)^{1/2} -2(C^\nu)^{-1}\alpha R_2\big(|\cdot|\big)\Big).
\end{split}\end{align}
Now we invoke the Cwikel-Lieb-Rozenblum inequality for $(-\Delta)^{1/2}$ in $\mathbb R^2$: By Remark 2.5 in \cite{Daubechies1983} (or Example 3.3 in \cite{Frank2014CLR}) there exists $C_{\text{CLR}} >0$ such that the right hand side of \eqref{R_2 game} does not exceed
\begin{align*}%\label{CLR constant}
 2C_{\text{CLR}}\big\|R_2\big(|\cdot|\big)\big\|_{\mathsf L^2(\mathbb R^2)}^2 =16\pi(C^\nu)^{-2}\alpha^2 C_{\text{CLR}}\int_0^\infty R_2^2(r)r\,\mathrm dr =:(\alpha/\alpha_1)^2.
\end{align*}
Thus for $\alpha\in [0, \alpha_1)$ the operator $D^{\nu, \kappa}_{\boldsymbol\theta(\kappa)}(2\alpha R_2\mathbb I)$ has no negative spectrum. Note that $\alpha_1$ does not depend on $\kappa$.

Letting
\begin{align*}
 \alpha_c :=\min\big(\{\alpha_0, \alpha_1\}\cup\{\alpha_\kappa: \kappa\in \mathbb Z+1/2, \ \kappa^2\leqslant \nu^2 +1/4\}\big) >0
\end{align*}
and combining the last result with \eqref{R_1 positivity} and \eqref{algebraic doubling} we observe that \eqref{D is positive} holds for all $\kappa\in \mathbb Z+1/2$ provided $\alpha\in [0, \alpha_c)$.
\end{proof}

\end{document}